\documentclass[12pt,draftcls,onecolumn]{IEEEtran}
\IEEEoverridecommandlockouts
\usepackage[cmex10]{amsmath}
\usepackage{relsize}
\usepackage{amssymb}
\usepackage{graphicx}
\usepackage{cite}
\usepackage{enumerate}
\graphicspath{{Figures/}}
\usepackage{epsfig}

\newtheorem{theorem}{Theorem}
\newtheorem{lemma}{Lemma}
\newtheorem{remark}{Remark}
\newtheorem{col}{Corollary}

\begin{document}
\title{A Generalized Expression for the Gradient of Mutual Information with the Application in Multiple Access Channels} 
\author{Mahboobeh Sedighizad, and Babak Seyfe,~\IEEEmembership{Senior Member,~IEEE}
\thanks{The authors are with the Information Theoretic Learning Systems Laboratory (ITLSL), Department of Electrical Engineering, Shahed University, Tehran, Iran (e-mail: \{m.sedighizad, Seyfe\}@shahed.ac.ir).}}

\maketitle\pagenumbering{arabic}\pagestyle{empty}\pagestyle{plain}

\begin{abstract}
Taking a functional approach, we derive a general expression for the gradient of the Mutual Information (MI) with respect to the system parameters in the stochastic systems. This expression covers the cases in which the system input depends on the system parameters. As an application, we consider the K-user Multiple Access Channels (MAC) with feedback and utilize the obtained results to explore the behavior of these systems in terms of the MI. Specializing the results to the additive Gaussian noise MAC, we extend the MI and Minimum Mean Square Error (MMSE) relationship, i.e., I-MMSE to the K-user Gaussian MAC with feedback. In this derivation, we show that the gradient of MI can be decomposed into three distinct parts, where the first part is the MMSE term originated from noise, and the second and third parts reflect the effects of the interference and feedback, respectively. Then, considering the capacity achieving Fourier-Modulated Estimate Correction (F-MEC) strategy of Kramer, we show how feedback compensates the destructive effects of the users' interference in the K-user symmetric Gaussian MAC.
\end{abstract} 


\section{Introduction}\label{Introduction}

The fundamental relationships between information theory and estimation theory attract the attentions of many researchers in recent
years~\cite{I-MMSE,vector,Palomar,guo2013interplay,Levy,venkat2015relations,Ghanem,EI-MMSE}. One of the major successful results in this area is the I-MMSE formula, which gives a new insight into the above theories~\cite{optimum,generalized,weighted,mmse2009,mimo,Transmitter}. 

The relation between the information and estimation theories was considered for the first time by Stam, where the derivative of the differential entropy was related to the Fisher information~\cite{Stam}. A connection between causal estimation error and input-output MI of a continuous-time additive white Gaussian noise channel, established by Duncan in~\cite{Duncan}. Kadota $et$ $al.$ in~\cite{Kadota} generalized the results reported by Duncan to the channels with feedback. 

Guo $et$ $al.$ in~\cite{I-MMSE} presented an explicit identity for the relationship between MMSE and the derivative of the input-output MI in the additive Gaussian noise channels. In~\cite{non-G} the relation between estimation theory and information theory was extended to the additive non-Gaussian noise channels. In~\cite{vector}, the I-MMSE relationship was generalized to the linear vector Gaussian channel. Similar connections between the derivative of MI with respect to the channel parameters and the error estimation were found for the continuous-time Poisson channel in~\cite{Poisson1} and~\cite{Poisson2}. In~\cite{Palomar}, Palomar and Verdu generalized the notion of the derivative of MI with respect to the system parameters to the arbitrary stochastic systems, where they assumed that the system input is independent of the system parameters. Their general formula is stated in a probabilistic setting in terms of the conditional marginal input distributions given the
outputs. In~\cite{pointwise}, a pointwise approach to generalize the above works is considered. In~\cite{Levy}, the relationship between the MI and the estimation error in scalar Levy channels as a more general class of observations model is expressed. The validity of I-MMSE relation for a fixed finite lookahead in a continuous-time additive white Gaussian noise channel is investigated in~\cite{venkat2013}. In~\cite{EI-MMSE} an extension of the I-MMSE relation has been presented for the discrete-time Gaussian channel with feedback. An extension of the I-MMSE formula to the additive Gaussian noise MAC without feedback is given in~\cite{Ghanem}. 

Most of the works which aim to give a general expression for the gradient of information measures, such as~\cite{Palomar} and~\cite{Levy}, are based on the statistical representation of the input-output relation of the system by use of the probabilistic descriptors. However, returning to the idea first described by Wiener, the input-output relation of a stochastic system can be described in an alternative manner by using a functional representation~\cite{Wiener,Network,li2018strong}. In~\cite{sedighizad2019gradient} and~\cite{sedighizad2019gradients} some of the achievements of this representation have been reported. 

Here, taking the functional approach we give a general expression for the gradient of MI with respect to the system parameters in a general system model with the application in the MAC. The main contributions of this paper can be summarized as follows,
\begin{enumerate}

\item We take a functional approach to represent the relationship between the input and output of a stochastic system in general. We suppose that the system model has been defined through the some known functions. We impose no constraints on the noise or the input distributions. Moreover, the system input is allowed to be a function of the system output and hence the system parameters, which may occur for example in the presence of feedback.  

\item We introduce the notion of MI variation in Theorem~\ref{Theorem General} and Corollary~\ref{Theorem 299}, which hold for any stochastic systems with continuous system input and system output. 

\item For the aforementioned system model, we obtain a general expression for the gradient of MI with respect to the system parameters in Theorem~\ref{Theorem 301}. 

 \item Particularizing the obtained results for MAC system with feedback as an important and practical system, the capability of the gradient of MI to interpret the behavior of this system is shown in Theorem~\ref{Theorem MACmain}. Corollaries~\ref{Corollary MAC2} and~\ref{Corollary MAC2_Prime} are devoted to the Gaussian case, where the extensions of I-MMSE formula are derived for the additive Gaussian noise MAC with feedback. As special cases of the obtained results, Corollaries~\ref{EX-IMMMSE},~\ref{Corollary MAC-Without3} and~\ref{Corollary Single-Without-Without} cover the results of~\cite{EI-MMSE},~\cite{Ghanem}, and~\cite{I-MMSE}, respectively. 

\item  In Section~\ref{K-user Gaussian MAC Kramer}, considering the capacity achieving Kramer's F-MEC code and using the proposed approach in this paper, we show how feedback compensates the negative effect of the user's interference on the capacity of a K-user symmetric Gaussian MAC.

\end{enumerate}

The reminder of this paper is organized as follows. Section~\ref{Definitions}, is devoted to the introduction of notation, definitions, assumptions, and system model. In section~\ref{Gradient of MI}, we first introduce the notion of MI variation, then we give a general expression for the gradient of MI with respect to the system parameters in a general system model. In section~\ref{Application}, we utilize the obtained results in a general MAC with and without feedback, where when the results specialized for the Gaussian case several extensions of I-MMSE formula are given. The single-user versions of the results are derived as well. In Section~\ref{K-user Gaussian MAC Kramer}, employing the F-MEC coding scheme we analyze the role of feedback in a K-user Gaussian MAC. Section~\ref{Conclusions} concludes the paper.

\section{Definitions, Assumptions and System Model}\label{Definitions}

In this section, we introduce notation, definitions, assumptions and the system model, which will be used in this paper.
\subsection{Notation}\label{Notation}
Scalar random variables are denoted by upper case letters, e.g., $X$, where a lower case letter $x$ is
used to denote a particular value of $X$; boldface lower case letters denote column random vectors, e.g., ${\textbf{x}}$, where a boldface italic lower case letter $\textbf{\textit {x}}$ is used to denote a particular value of $\textbf{x}$. Boldface upper case letters denote random matrixes, e.g., $\bf A$, where boldface italic upper case letter $\textit{\textbf{A}}$ is used to denote a particular value of $\bf A$. Also, $\left(  \cdot  \right)_{i,j}$ and $\left(  \cdot  \right)_l$ denote the ($i$th,\,\,$j$th) element of a matrix and $l$th element of a vector, respectively. Character ${\bf{I}}$ denotes the identity matrix, the superscript $(\cdot)^T$ denotes the transpose, and $\left\|  \cdot  \right\|$ denotes the norm of a vector.

Gradient of the scalar function $f$ taken with respect to vector ${\boldsymbol{\gamma }} = \left( {\gamma _1 ,\gamma _2 ,...,\gamma _J } \right)^T$
is denoted by $\nabla _{{\boldsymbol{\gamma }}} f$, which is defined as,
\begin{eqnarray}\label{eqnGradient}
\nabla _{\boldsymbol{\gamma }} f = \left( {\frac{{\partial f}}{{\partial \gamma _1 }},...,\frac{{\partial f}}{{\partial \gamma _J }}} \right)^T.
\end{eqnarray}
Derivative of
 the vector function ${\bf{f}} = \left( {f_1 ,f_2 ,...,f_n } \right)^T$ with respect to the vector
 ${\boldsymbol{\gamma }} $ is defined as,
\begin{eqnarray}\label{eqnDeriva}
\frac{{\partial {\bf{f}}}}{{\partial {\boldsymbol{\gamma }}}} = \left[ \begin{array}{l}
 \begin{array}{*{20}c}
   {\frac{{\partial f_1 }}{{\partial \gamma _1 }}} & {\begin{array}{*{20}c}
    \cdots  & {\frac{{\partial f_1 }}{{\partial \gamma _J }}}  \\
\end{array}}  \\
\end{array} \\
 \;\;\begin{array}{*{20}c}
   { \vdots \;\;\;\;} &  \ddots   \\
\end{array}\;\;\; \vdots  \\
 \begin{array}{*{20}c}
   {\frac{{\partial f_n }}{{\partial \gamma _1 }}} & {\begin{array}{*{20}c}
    \cdots  & {\frac{{\partial f_n }}{{\partial \gamma _J }}}  \\
\end{array}}  \\
\end{array} \\
 \end{array} \right],
\end{eqnarray}
and derivative of vector function $\bf f$ with respect to the scalar $\gamma_j$ is,
\begin{eqnarray}\label{eqnScalar-derivative}
\frac{{\partial {\textit{\bf{f}}}}}{{\partial \gamma _j }} = \left( {\frac{{\partial f_1 }}{{\partial \gamma _j }},...,\frac{{\partial f_n }}{{\partial \gamma _j }}} \right)^T.
\end{eqnarray}

For deterministic vector ${\boldsymbol{\varepsilon }}$ and scalar function $r\left(  \cdot  \right)$ we say that $r\left( {\boldsymbol{\varepsilon }} \right) \buildrel \Delta \over = o\left( {\boldsymbol{\varepsilon }} \right)$, if $\mathop {\lim }\limits_{{\boldsymbol{\varepsilon }} \to {\bf{0}}} \frac{{r\left( {\boldsymbol{\varepsilon }} \right)}}{{\left\| {\boldsymbol{\varepsilon }} \right\|}} = 0$ and the Gaussian capacity function is denoted by $\texttt{C}\left( \chi \right) = \left( {1 \mathord{\left/
 {\vphantom {1 2}} \right.
 \kern-\nulldelimiterspace} 2}\right)\ln \left( {1 + \chi} \right)$ for $\chi \ge 0$.
  \subsection{Definitions}\label{Definition 2}
 The vector function ${\boldsymbol {\phi }}_{\bf{q}} \left( \textbf{\textit {q}} \right)$,
is the Joint Score Function (JSF) of the random vector ${\textbf{q}}$ which is defined as the
log-derivative of its probability density function~\cite{16}, i.e.,
\begin{eqnarray}\label{eqn4000}
\boldsymbol \phi _{\bf{q}} \left( \textbf{\textit {q}} \right) = \nabla _{\textbf{\textit {q}}} \ln p_{\bf{q}} \left( \textbf{\textit {q}} \right)
\end{eqnarray}
where $p_{\bf{q}} \left( \textbf{\textit {q}} \right)$ is the joint pdf of the random vector ${\textbf{q}}$.
We define $\boldsymbol \phi _{{\textbf x, \textbf y}}^{\textbf{x}} \left( \textit{\textbf x}, \textit{\textbf y} \right)$ and $\boldsymbol \phi _{{\textbf x, \textbf y}}^{\textbf{y}} \left( \textit{\textbf x}, \textit{\textbf y} \right)$
as,
\begin{eqnarray}\label{eqn300}
\boldsymbol \phi _{{\textbf x, \textbf y}}^{\textbf{x}} \left( \textit{\textbf x}, \textit{\textbf y} \right) \buildrel \Delta \over =\nabla _{\textbf{\textit {x}}} \ln p_{{\textbf{x}},{\textbf{y}}} \left( \textbf{\textit {x}},\textbf{\textit {y}} \right)
\end{eqnarray}
and
\begin{eqnarray}\label{eqn15}
\boldsymbol \phi _{{\textbf x, \textbf y}}^{\textbf{y}} \left( \textit{\textbf x}, \textit{\textbf y} \right)\buildrel \Delta \over =\nabla _{\textbf{\textit {y}}} \ln p_{{\textbf{x}},{\textbf{y}}} \left( \textbf{\textit {x}},\textbf{\textit {y}} \right)
\end{eqnarray}
where $p_{{\textbf{x}},{\textbf{y}}} \left( \textbf{\textit {x}},\textbf{\textit {y}} \right)$ is the
joint pdf of the random vectors ${\textbf{x}}$ and ${\textbf{y}}$.
Moreover, the vector function $\boldsymbol \phi _{{\textbf x, \textbf y}}\left( \textit{\textbf x}, \textit{\textbf y} \right) $, is defined as,
\begin{eqnarray}\label{eqn10}
\boldsymbol \phi _{{\textbf{x,y}}} \left( \textbf{\textit{x}},\textbf{\textit{y}} \right) \buildrel \Delta \over = \left( {\boldsymbol \phi _{{\textbf{x,y}}}^{{\textbf{x}}\;\;T} \left( \textbf{\textit{x}},\textbf{\textit{y}} \right),\boldsymbol \phi _{{\textbf{x,y}}}^{{\textbf{y}}\;\;T} \left(\textbf{\textit{x}},\textbf{\textit{y}} \right)} \right)^T.
\end{eqnarray}
\subsection{Assumptions}\label{Assumptions}
Throughout the paper we consider the following assumptions. The random variables are continuous with the following
 definition. Random variable $X$ with cumulative distribution function $F\left( x \right) = \Pr \left( {X \le x} \right)$ is said to be continuous if $F\left( x \right)$ is
 continuous~\cite{16}. As~\cite{Pham}, we assume that the score functions are bounded in the sense that,
\begin{eqnarray}\label{eqnScore_Bound}
\left\| {{\boldsymbol{\phi }}_{\textbf{q}} \left( {\textit{\textbf{q}}} \right)} \right\| \le C\left( {1 + \left\| {\textit{\textbf{q}}} \right\|^{\alpha  - 1} } \right),\,\,\,\ {\rm{for\, all}}\, \textit{\textbf q}
\end{eqnarray}
for some constant $C$ and $\alpha\geq1$.
\subsection{System Model}\label{System Model}
Taking a functional approach, we describe the stochastic systems with the following general system model,
\begin{eqnarray}\label{eqn11}
{\textbf{y}} = {\bf{f}}\left( {{\textbf{x}},{\textbf{w}},\boldsymbol\gamma } \right);\;\; \boldsymbol{\gamma}  \in \Theta
\end{eqnarray}
and
\begin{eqnarray}\label{eqn301}
{\textbf{x}} = {\bf g}(\textbf {u},\textbf {y})
\end{eqnarray}
where the random vectors $\textbf x$ and ${\textbf{y}}$ are the system input and system output, respectively. ${\bf f}$ and ${\bf g}$ are continuously differentiable known functions with bounded partial derivatives. Arbitrary but known-distribution random vector ${\textbf{w}}$ stands for any unwanted processes which can affect the system. $\boldsymbol \gamma$ is an arbitrary finite-dimensional parameter from the deterministic system parameters set $\Theta$.
As it can be seen from~\eqref{eqn301}, in this general system model the system input is allowed to be dependent on the system output and hence the system parameters, which may
happen for example in the systems with feedback. Random vector $\textbf {u}$ can stand for the sequence of the source messages assumed to
 be independent of $\textbf w$ and the set of system parameters.
\section{Gradient of Mutual Information }\label{Gradient of MI}

As it is shown in~\cite{optimum,generalized,weighted,mmse2009,mimo,Transmitter}, the gradient of MI plays a central role to give a new insight into the some fundamental notions in information theory and also it can be considered as an important measure in optimization problems. However, this pivotal role does not confined to the aforementioned items as we see in the next section.

In this section, we first introduce the notion of variation of MI between two arbitrary random vectors $\textbf x$ and $\textbf y$. Then, using the introduced notion of the MI variation, the gradient of MI with respect to
the system parameters in the system model defined by~\eqref{eqn11} and~\eqref{eqn301} will be given.
\subsection{Variation of Mutual Information  }\label{Variation of Mutual Information}
Here, we introduce a general expression for the variation of MI. Let $\Delta I\left( {{\textbf{x}};{\textbf{y}}} \right) = I\left( {{\textbf{x}} + {\boldsymbol{\delta }}_{\textbf{x}} ;{\textbf{y}} + {\boldsymbol{\delta }}_{\textbf{y}} } \right) - I\left( {{\textbf{x}};{\textbf{y}}} \right)$ be the variation of mutual
information between random vectors $\textbf x$ and $\textbf y$ caused by random variations
\begin{eqnarray}\label{eqnDelta_x}
{\boldsymbol{\delta }}_{\textbf{x}}  = {\bf{\tilde X \boldsymbol {\varepsilon} }} + o\left( {\boldsymbol{\varepsilon }} \right)
\end{eqnarray}
and
\begin{eqnarray}\label{eqnDelta_y}
{\boldsymbol{\delta }}_{\textbf{y}}  = {\boldsymbol {\tilde Y\varepsilon }} + o\left( {\boldsymbol{\varepsilon }} \right)
\end{eqnarray}
where ${\boldsymbol{\tilde X}}$ and ${\boldsymbol{\tilde Y}}$ are two random matrixes, and $\boldsymbol{\varepsilon }$ is a deterministic vector for which the products ${\boldsymbol {\tilde X\varepsilon }}$ and ${\boldsymbol {\tilde Y\varepsilon }}$ make sense and have the same dimensions as $\textbf{x}$ and $\textbf{y}$, respectively. Then, the following Theorem holds.
 \begin{theorem}\label{Theorem General}
 Variation of MI between $\textbf x$
and $\textbf y$ resulting from ${\boldsymbol{\delta }}_{\textbf{x}}$ and ${\boldsymbol{\delta }}_{\textbf{y}}$
 defined by~\eqref{eqnDelta_x} and~\eqref{eqnDelta_y}, as ${\boldsymbol{\varepsilon }} \to {\bf{0}}$ is,
\begin{align}\label{eqnGeneral}
\Delta I \left( {{\textbf{x}};{\textbf{y}}} \right)= {\mathop{\rm E}\nolimits} \left\{ {\left( {\boldsymbol \phi _{{\textbf{x},\textbf{y}}}^{\textbf{x}} \left( {{{\textbf{x},\textbf{y}}}} \right) - \boldsymbol \phi _{\textbf{x}} \left( {{\textbf{x}}} \right)} \right)^T {\bf{\tilde X\boldsymbol \varepsilon }} } \right\} + {\mathop{\rm E}\nolimits} \left\{ {\left( {\boldsymbol \phi _{{\textbf{x},\textbf{y}}}^{\textbf{y}} \left( {{{\textbf{x},\textbf{y}}}} \right) - \boldsymbol \phi _{\textbf{y}} \left( {{\textbf{y}}} \right)} \right)^T {\bf{\tilde Y\boldsymbol \varepsilon }} } \right\} + o\left( {\boldsymbol{\varepsilon }} \right).
  \end{align}
 \end{theorem}
\begin{proof}
See Appendix~\ref{proof Theorem General}.
\end{proof}

 Now, we use the general expression of~\eqref{eqnGeneral} in the general system model defined by~\eqref{eqn11} and~\eqref{eqn301}, with $\left( {M \times 1} \right)$ vector $\textbf x$ and $\left( {N \times 1} \right)$ vector $\textbf y$ as the system input and system output, respectively. Without loss of generality suppose that we are interested in the variation of the MI caused by a small variation in a specific element of ${\Theta }$, as $\left( {J \times 1} \right)$ vector ${\boldsymbol {\gamma }} $. Let ${\boldsymbol{\hat \gamma }} = {\boldsymbol{\gamma }} + {\boldsymbol{\varepsilon }}_{\boldsymbol{\gamma }}$, where ${\boldsymbol{\varepsilon }}_{\boldsymbol{\gamma }}$ is a deterministic vector with the same dimension as ${\boldsymbol {\gamma }}$. Then, the next Corollary expresses the variation of the MI between $\textbf x$ and $\textbf y$ resulting from small variation in the system parameter ${\boldsymbol {\gamma }}$.
  \begin{col}\label{Theorem 299}
 Considering the system model~\eqref{eqn11} and~\eqref{eqn301}, variation of the MI between $\textbf x$ and $\textbf y$ resulting from ${\boldsymbol{\varepsilon }}_{\boldsymbol{\gamma }}$
 as ${\boldsymbol{\varepsilon }}_{\boldsymbol{\gamma }}\rightarrow \bf 0$, is
\begin{align}\label{eqn1}
\Delta I \left( {{\textbf{x}};{\textbf{y}}} \right)= {\mathop{\rm E}\nolimits} \left\{ {\left( {\boldsymbol \phi _{{\textbf x, \textbf y}}^{\textbf{x}} \left( {\textbf x}, {\textbf y} \right) - \boldsymbol \phi _{\textbf{x}} \left( {{\textbf{x}}} \right)} \right)^T \frac{{\partial {\bf{g}}}}{{\partial {\boldsymbol{\gamma }}}}{\boldsymbol{\varepsilon }}_{\boldsymbol{\gamma }} } \right\} + {\mathop{\rm E}\nolimits} \left\{ {\left( {\boldsymbol \phi _{{\textbf x, \textbf y}}^{\textbf{y}} \left( {\textbf x}, {\textbf y} \right) - \boldsymbol \phi _{\textbf{y}} \left( {{\textbf{y}}} \right)} \right)^T \frac{{\partial {\bf{f}}}}{{\partial {\boldsymbol{\gamma }}}}{\boldsymbol{\varepsilon }}_{\boldsymbol{\gamma }} } \right\} + o\left( {\boldsymbol{\varepsilon }}_{\boldsymbol{\gamma }} \right)
  \end{align}
 \end{col}
 where we have written ${\bf{f}}\left( {{\textbf{x}},{\textbf{w}},{\boldsymbol{\gamma }}} \right)$ and ${\bf{g}}\left( {{\textbf{u}},{\textbf{y}}} \right)$ simply as ${\bf{f}}$ and ${\bf{g}}$, and ${{\partial {\bf{g}}} \mathord{\left/
 {\vphantom {{\partial {\bf{g}}} {\partial {\boldsymbol \gamma} }}} \right.
 \kern-\nulldelimiterspace} {\partial {\boldsymbol \gamma} }}$ and ${{\partial {\bf{f}}} \mathord{\left/
 {\vphantom {{\partial {\bf{f}}} {\partial {\boldsymbol \gamma} }}} \right.
 \kern-\nulldelimiterspace} {\partial {\boldsymbol \gamma} }}$ are matrixes of the partial derivatives of $\bf g$ and $\bf f$ with respect to the system parameter $\boldsymbol \gamma$, respectively.
  \begin{proof}
Let ${\boldsymbol{\delta }}_{\textbf x} $ and ${\boldsymbol{\delta }}_{\textbf y} $
denote the variations of $\textbf x$ and $\textbf y$ caused by
${\boldsymbol{\varepsilon }}_{\boldsymbol{\gamma }}$. Then, based on~\eqref{eqn11} and~\eqref{eqn301}, and by Taylor series expansion of $\bf g$ and $\bf f$ around $\boldsymbol{\gamma }$
we can write,
\begin{align}\label{eqn22223}
{\boldsymbol{\delta }}_{\textbf{x}} & = {\bf{g}}\left( { \cdot ,{\boldsymbol{\gamma }} + {\boldsymbol{\varepsilon }}_{\boldsymbol{\gamma }} } \right) - {\bf{g}}\left( { \cdot ,{\boldsymbol{\gamma }}} \right)\notag\\
&= \frac{{\partial {\bf{g}}}}{{\partial {\boldsymbol{\gamma }}}}{\boldsymbol{\varepsilon }}_{\boldsymbol{\gamma }} + o\left( {\boldsymbol{\varepsilon }}_{\boldsymbol{\gamma }} \right)
\end{align}
and
\begin{align}\label{eqn22224}
{\boldsymbol{\delta }}_{\textbf{y}}  &= {\bf{f}}\left( { \cdot ,\cdot,{\boldsymbol{\gamma }} + {\boldsymbol{\varepsilon }}_{\boldsymbol{\gamma }} } \right) - {\bf{f}}\left( { \cdot ,\cdot,{\boldsymbol{\gamma }}} \right)\notag\\
& =\frac{{\partial {\bf{f}}}}{{\partial {\boldsymbol{\gamma }}}}{\boldsymbol{\varepsilon }}_{\boldsymbol{\gamma }}  + o\left( {\boldsymbol{\varepsilon }}_{\boldsymbol{\gamma }} \right).
\end{align}
Comparing~\eqref{eqn22223} and~\eqref{eqn22224} with the definitions of ${\boldsymbol{\delta }}_{\textbf{x}}$ and ${\boldsymbol{\delta }}_{\textbf{y}}$ in Theorem~\ref{Theorem General}, ${{\partial {\bf{g}}} \mathord{\left/
 {\vphantom {{\partial {\bf{g}}} {\partial {\boldsymbol{\gamma }}}}} \right.
 \kern-\nulldelimiterspace} {\partial {\boldsymbol{\gamma }}}}$ and ${{\partial {\bf{f}}} \mathord{\left/
 {\vphantom {{\partial {\bf{f}}} {\partial {\boldsymbol{\gamma }}}}} \right.
 \kern-\nulldelimiterspace} {\partial {\boldsymbol{\gamma }}}}$ can be considered as ${\bf{\tilde X}}$ and ${\bf{\tilde Y}}$, and ${\boldsymbol{\varepsilon }}_{\boldsymbol{\gamma }}$ as $\boldsymbol \varepsilon$ respectively which completes the proof.
\end{proof}
\begin{remark}\label{Remark 2}
  Substituting the definitions of the JSFs from Section~\ref{Definition 2} in~\eqref{eqn1} results in an
   alternative expression as,
\begin{align}\label{eqn304}
 \Delta I\left( {{\textbf{x}};{\textbf{y}}} \right) = {\mathop{\rm E}\nolimits} \left\{ {\nabla _{{\textbf{x}}}^T \ln p_{{\textbf{y}}|{\textbf{x}}} \left( {{{\textbf{y}}}|{{\textbf{x}}}} \right) \frac{{\partial {\bf{g}}}}{{\partial {\boldsymbol{\gamma }}}}{\boldsymbol{\varepsilon }}_{\boldsymbol{\gamma }} } \right\} + {\mathop{\rm E}\nolimits} \left\{ {
\nabla _{{\textbf{y}}}^T \ln p_{{\textbf{x}}|{\textbf{y}}} \left( {{{\textbf{x}}}|{{\textbf{y}}}} \right) \frac{{\partial {\bf{f}}}}{{\partial {\boldsymbol{\gamma }}}}{\boldsymbol{\varepsilon }}_{\boldsymbol{\gamma }} } \right\} + o\left( {\boldsymbol{\varepsilon }}_{\boldsymbol{\gamma }} \right).
 \end{align}
 \end{remark}
\begin{remark}\label{Remark 20}
From~\eqref{eqn304} it can be seen that, if the system input is not a function of the system parameter of interest, then ${{\partial {\bf{g}}} \mathord{\left/
 {\vphantom {{\partial {\bf{g}}} {\partial {\bf{\gamma }}}}} \right.
 \kern-\nulldelimiterspace} {\partial {\boldsymbol{\gamma }}}} = \bf0$ and hence by the Lebesgue dominated convergence Theorem~\cite{Rudin}, and our assumption about the boundedness of score functions in~\ref{Assumptions}, the first term of this equation will vanish,
\begin{align}\label{eqnShort}
 \Delta I\left( {{\textbf{x}};{\textbf{y}}} \right) =  {\mathop{\rm E}\nolimits} \left\{ {
\nabla _{{\textbf{y}}}^T \ln p_{{\textbf{x}}|{\textbf{y}}} \left( {{{\textbf{x}}}|{{\textbf{y}}}} \right) \frac{{\partial {\bf{f}}}}{{\partial {\boldsymbol{\gamma }}}}{\boldsymbol{\varepsilon }}_{\boldsymbol{\gamma }} } \right\} + o\left( {\boldsymbol{\varepsilon }}_{\boldsymbol{\gamma }} \right).
 \end{align}
\end{remark}

 The above Corollary and Remarks enable us to calculate the gradient of MI with
 respect to any system parameters of interest which will be given in the next subsection.
\subsection{Gradient of Mutual Information With Respect to the System Parameters}\label{Gradient of Mutual Information With Respect to the System Parameters}

In this section, we give a general expression for the gradient of MI with respect to the system parameters, by considering the system
model as the general form stated in~\ref{System Model}. We assume that, both input and output of the system are allowed to be affected by the system parameters. Hence, this general
model covers many practical system models including stochastic systems with feedback.
\begin{theorem}\label{Theorem 301}
Consider the system model introduced by~\eqref{eqn11} and~\eqref{eqn301}. Then, the gradient of MI with respect to a specific vector from the parameter
set $ \Theta$ such as $\boldsymbol \gamma=(\gamma_1,...,\gamma_J)^T$, will be as,
\begin{align}\label{eqn17}
 \nabla _{\boldsymbol \gamma}  I\left( {{\textbf{x}};{\textbf{y}}} \right)& = {\mathop{\rm E}\nolimits} \left\{ {\left( {\frac{{\partial {\bf{g}}}}{{\partial {\boldsymbol \gamma} }}} \right)^T \left( {\boldsymbol \phi _{{\textbf x, \textbf y}}^{\textbf{x}} \left( {\textbf x}, {\textbf y} \right) - \boldsymbol \phi _{\textbf{x}} \left( {{\textbf{x}}} \right)} \right)} \right\} \notag \\
 & \,\,\,\,\, + {\mathop{\rm E}\nolimits} \left\{ {\left( {\frac{{\partial {\bf{f}}}}{{\partial {\boldsymbol \gamma} }}} \right)^T \left( {\boldsymbol \phi _{{\textbf x, \textbf y}}^{\textbf{y}} \left( {\textbf x}, {\textbf y} \right) - \boldsymbol \phi _{\textbf{y}} \left( {{\textbf{y}}} \right)} \right)} \right\}
\end{align}
\end{theorem}
\begin{proof}
Expanding~\eqref{eqn1} we have,
\begin{align}\label{eqn21}
 \Delta I\left( {{\textbf{x}};{\textbf{y}}} \right) &= \sum\limits_{j = 1}^J {\left( {\sum\limits_{m = 1}^M {{\mathop{\rm E}\nolimits} \left\{ {\left( {\frac{{\partial {\bf{g}}}}{{\partial {\boldsymbol {\gamma }}}}} \right)_{m,j} \left( {\boldsymbol \phi _{{\textbf x, \textbf y}}^{\textbf{x}} \left( {\textbf x}, {\textbf y} \right) - \boldsymbol \phi _{\textbf{x}} \left( {{\textbf{x}}} \right)} \right)_m } \right\}} } \right.} \notag \\
& \,\,\,\,\, \left. { + \sum\limits_{n = 1}^N {{\mathop{\rm E}\nolimits} \left\{ {\left( {\frac{{\partial {\bf{f}}}}{{\partial {\boldsymbol {\gamma }}}}} \right)_{n,j} \left( {\boldsymbol \phi _{{\textbf x, \textbf y}}^{\textbf{y}} \left( {\textbf x}, {\textbf y} \right) - \boldsymbol \phi _{\textbf{y}} \left( {{\textbf{y}}} \right)} \right)_n } \right\}} } \right)\left( {\boldsymbol{\varepsilon }}_{\boldsymbol{\gamma }} \right)_j + o\left( {\boldsymbol{\varepsilon }}_{\boldsymbol{\gamma }} \right).
\end{align}
This equation shows that as $\left({\boldsymbol{\varepsilon }}_{\boldsymbol{\gamma }}\right)_j\rightarrow 0$,
\begin{align}\label{eqn306}
  \frac{\partial }{{\partial \gamma _j }}I\left( {{\textbf{x}};{\textbf{y}}} \right) &= \sum\limits_{m = 1}^M {{\mathop{\rm E}\nolimits} \left\{ {\left( {\frac{{\partial {\bf{g}}}}{{\partial {\boldsymbol \gamma} }}} \right)_{m,j} \left( {\boldsymbol \phi _{{\textbf x, \textbf y}}^{\textbf{x}} \left( {\textbf x}, {\textbf y} \right) - \boldsymbol \phi _{\textbf{x}} \left( {{\textbf{x}}} \right)} \right)_m } \right\}}  \notag \\
  & \,\,\,\,\, + \sum\limits_{n = 1}^N {{\mathop{\rm E}\nolimits} \left\{ {\left( {\frac{{\partial {\bf{f}}}}{{\partial {\boldsymbol \gamma} }}} \right)_{n,j} \left( {\boldsymbol \phi _{{\textbf x, \textbf y}}^{\textbf{y}} \left( {\textbf x}, {\textbf y} \right) - \boldsymbol \phi _{\textbf{y}} \left( {{\textbf{y}}} \right)} \right)_n } \right\}}\notag \\
    &= {\mathop{\rm E}\nolimits} \left\{ {\left( {\frac{{\partial {\bf{g}}}}{{\partial \gamma _j }}} \right)^T \left( {{\boldsymbol{\phi }}_{{\textbf{x,y}}}^{\textbf{x}} \left( {{{\textbf{x}}},{{\textbf{y}}}} \right){\boldsymbol{ - \phi }}_{\textbf{x}} \left( {\textbf{x}} \right)} \right)} \right\} \notag\\
 &\,\,\,\,\, + {\mathop{\rm E}\nolimits} \left\{ {\left( {\frac{{\partial {\bf{f}}}}{{\partial \gamma _j }}} \right)^T \left( {{\boldsymbol{\phi }}_{{\textbf{x,y}}}^{\textbf{y}} \left( {{{\textbf{x}}},{{\textbf{y}}}} \right){\boldsymbol{ - \phi }}_{\textbf{y}} \left( {\textbf{y}} \right)} \right)} \right\}
\end{align}
where, ${{\partial {\bf{g}}} \mathord{\left/
 {\vphantom {{\partial {\bf{g}}} {\partial \gamma _j }}} \right.
 \kern-\nulldelimiterspace} {\partial \gamma _j }}$ and ${{\partial {\bf{f}}} \mathord{\left/
 {\vphantom {{\partial {\bf{f}}} {\partial \gamma _j }}} \right.
 \kern-\nulldelimiterspace} {\partial \gamma _j }}$ are derivatives of vector functions $\bf f$ and $\bf g$ with respect to $\gamma_j$, respectively. Regarding to the definition of the gradient of a scalar function with respect to a vector,~\eqref{eqn306} completes the proof.
\end{proof}
\begin{remark}\label{Remark 3}
Utilizing the definitions of JSFs in~\ref{Definition 2}, an alternative expression for Theorem~\ref{Theorem 301} can be written as,
\begin{align}\label{eqn305}
 \nabla _{\boldsymbol \gamma}  I\left( {{\textbf{x}};{\textbf{y}}} \right)& = {\mathop{\rm E}\nolimits} \left\{ {\left( {\frac{{\partial {\bf{g}}}}{{\partial {\boldsymbol \gamma} }}} \right)^T \nabla _{{\textbf{x}}} \ln p_{{\textbf{y}}\left| {\textbf{x}} \right.} \left( {{{\textbf{y}}}\left| {{\textbf{x}}} \right.} \right)} \right\} \notag \\
 & \,\,\,\,\, + {\mathop{\rm E}\nolimits} \left\{ {\left( {\frac{{\partial {\bf{f}}}}{{\partial {\boldsymbol \gamma} }}} \right)^T \nabla _{{\textbf{y}}} \ln p_{{\textbf{x}}\left| {\textbf{y}} \right.} \left( {{{\textbf{x}}}\left| {{\textbf{y}}} \right.} \right)} \right\}.
  \end{align}
  \end{remark}
Again, if the system input is not a function of the system parameters, then the first term of~\eqref{eqn305} will vanish.

  In the following section, we use our general results in multiple access communication channel with feedback, as one of the realistic scenarios in which the system input
  depends on the system parameters in general.
\section{Applications in Multiple Access Communication Channels}\label{Application}
The general results obtained in the previous section are applicable to any stochastic system defined by~\eqref{eqn11} and~\eqref{eqn301}. In this section, we particularize these general results for the multiuser communication channel as one of the most important practical stochastic systems. Specifically, we consider the multiple access channel with feedback, which has applications in the cellular networks, medium access in a Local Area Networks (LAN), etc. We first introduce a general model for MAC channel with feedback and then calculate the sensitivity of the MI to the channel parameters in this model. These results are particularized to the additive noise MAC and additive Gaussian noise MAC cases. We also explore the counterparts of the results for the MAC without
feedback. The single-user versions of the results are given as well. 
\subsection{Channels with Feedback}\label{Feedback}
In this section, we study the sensitivity of the MI with respect to the parameters in a communication channel with feedback. We first introduce a general model for MAC with feedback
 and then we consider a single-user channel with feedback as a special case.
\subsubsection{MAC with Feedback}\label{MAC-Feedback}
   In this section, we consider MAC with feedback as a familiar and practical scenario of multiuser
 communication systems. We first give a general expression for the gradient of MI with respect to the channel parameters in a MAC with arbitrary channel model. At the second step, we consider an additive Gaussian noise MAC with feedback, where an extension of I-MMSE relationship is given. 

Fig.~\ref{fig:MAC} shows a schematic view of the MAC with noiseless causal feedbacks from the receiver to all $K$ transmitters.
\begin{figure}
  \centering
  \includegraphics[width=15cm]{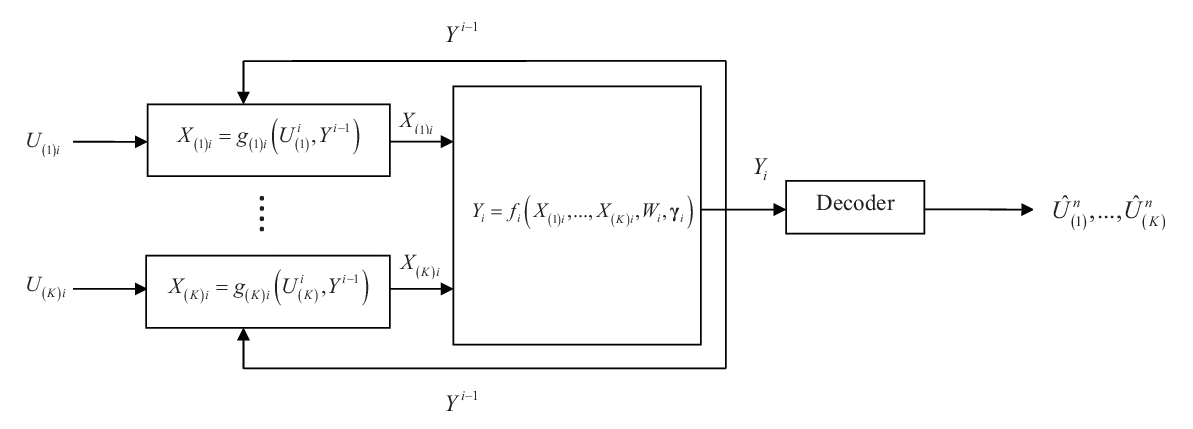}
 \caption{Discrete-time multiple access channel with feedback.}
 \label{fig:MAC}
\end{figure}
The $k{\rm{th}}$ transmitter $\left(k=1,...,K\right)$, wishes to
communicate symbol $M_{(k)}$ to the receiver by sending
$U_{_{\left( k \right)} }^n$ in $n$ uses of the channel,
where $U_{\left( k \right)}^i$ stands for the sequence of messages of
$k$th user up to time instant $i,\, i\in[1:n]$. A $\left(2^{nR_1},...,2^{nR_K},n \right)$ code for the MAC with feedback consists of $K$ message sets $\left[1:2^{nR_k}\right]$, $k=1,...,K$, $K$ encoders as,
\begin{eqnarray}\label{eqnMAC2}
X_{\left( k \right)i}  = g_{\left( k \right)i} \left( {U_{\left( k \right)}^i ,Y^{i - 1} } \right);\;\;k = 1,...,K,\;i \in \left[ {1:n} \right]
\end{eqnarray}
and a decoder $\left( {\hat M_{(1)} \left( {Y^n } \right),...,\hat M_{(K)} \left( {Y^n } \right)} \right)$, where $X_{\left( k \right)i}$
and $Y^{i-1}$ are used for the $k{\rm{th}}$ channel input at time instant $i$ and the
channel outputs up to time instant $i-1$, respectively. We consider the
following general functional model for the channel
in $i{\rm{th}}$ $\left(i=1,...,n\right)$ use of the channel,
\begin{eqnarray}\label{eqnMAC1}
Y_i  = f_i \left( {X_{\left( 1 \right)i} ,...,X_{\left( K \right)i} ,W_i ,\boldsymbol\gamma_i } \right),\;\;i \in \left[ {1:n} \right]
\end{eqnarray}
where, $W_i$ stands for arbitrary but known distribution noise and ${\boldsymbol {\gamma }_i} = \left( {\gamma _{(1)i} ,...,\gamma _{(K)i} } \right)^T$ denotes the channel parameters vector, where we are interested to find the sensitivity of MI with respect to it. The following Theorem gives a general expression for this issue. 
\begin{theorem}\label{Theorem MACmain}
 In a general MAC with noiseless causal feedback as described in~\eqref{eqnMAC2} and~\eqref{eqnMAC1}, gradient of MI
with respect to the system parameter $\boldsymbol \gamma _j ,\;j = 1,...,n$ is,
\begin{align}\label{eqnMAC3}
\nabla _{{\boldsymbol{\gamma }}_j }I\left( {X_{\left( 1 \right)}^n ,...,X_{\left( K \right)}^n ;Y^n } \right)& =  \sum\limits_{k = 1}^K {{\mathop{\rm E}\nolimits} \left\{ {\left( \frac{{\partial {\bf{g}}_{\left( k \right)} }}{{\partial \boldsymbol\gamma_j }}\right)^T\nabla _{X_{\left(k \right)}^n } \ln p_{Y^n \left| {X_{\left( 1 \right)}^n ,...,X_{\left( K \right)}^n } \right.} \left( {Y^n \left| {X_{\left( 1 \right)}^n ,...,X_{\left( K \right)}^n } \right.} \right)} \right\}}\notag \\
&\,\,\,\,\,+  {\mathop{\rm E}\nolimits} \left\{ \left( \frac{{\partial {\bf{f}}}}{{\partial \boldsymbol\gamma_j }}\right)^T{\nabla _{Y^n } \ln p_{X_{\left( 1 \right)}^n ,...,X_{\left( K \right)}^n \left| {Y^n } \right.} \left( {X_{\left( 1 \right)}^n ,...,X_{\left( K \right)}^n \left| {Y^n } \right.} \right)} \right\}
\end{align}
where, $X_{\left( k \right)}^n$ stands for the sequence of the channel inputs of $k{\rm{th}}$ user and $Y^n$ denotes the sequence of channel outputs in $n$ uses of the channel, respectively, ${\bf{f}} = \left( {f_1 ,...,f_n } \right)^T$ with $f_i  = f_i \left( {X_{\left( 1 \right)i} ,...,X_{\left( K \right)i} ,W_i ,\boldsymbol\gamma_i } \right)$, and ${\bf{g}}_{\left( k \right)}  = \left( {g_{\left( k \right)1} ,...,g_{\left( k \right)n} } \right)^T$ with
$g_{\left( k \right)i}  = g_{\left( k \right)i}  \left( {U_{\left( k \right)}^i ,Y^{i - 1} } \right)$. Also,
\begin{eqnarray}\label{eqnMAC4}
\nabla _{{\boldsymbol{\gamma }}_j }I\left( {U_{\left( 1 \right)}^n ,...,U_{\left( K \right)}^n ;Y^n } \right)= {\mathop{\rm E}\nolimits} \left\{ \left(\frac{{\partial {\bf{f}}}}{{\partial \boldsymbol\gamma_j }} \right)^T{\nabla _{Y^n } \ln p_{U_{\left( 1 \right)}^n ,...,U_{\left( K \right)}^n \left| {Y^n } \right.} \left( {U_{\left( 1 \right)}^n ,...,U_{\left( K \right)}^n \left| {Y^n } \right.} \right)} \right\}.
\end{eqnarray}
\end{theorem}
\begin{proof}
Proof of~\eqref{eqnMAC3} directly follows from Remark~\ref{Remark 3} by considering ${\textbf{x}} = \left( {X_{\left( 1 \right)}^n ,...,X_{\left( K \right)}^n } \right)^T$ and ${\textbf{y}} = Y^n$.

For~\eqref{eqnMAC4}, considering the fact that the sequence of messages are not dependent on the system parameters, Remark~\ref{Remark 3} completes the proof.
\end{proof}
\begin{remark}\label{Remark M_U}
It should be noted that~\eqref{eqnMAC4} can be considered as $\nabla _{{\boldsymbol{\gamma }}_j } I\left( {M_{\left( 1 \right)} ,...,M_{\left( K \right)} ;Y^n } \right)$ as well. Because, message $M_{\left( k \right)} ,\;\;k \in \left[ {1:K} \right]$ uniquely determines sequence $U_{\left( k \right)}^n$, and vice versa.
\end{remark}

Here, we particularize our results to a general additive noise MAC with $K$ users in the presence of
feedback.
 \begin{col}\label{Corollary MAC1}
In an additive noise MAC with feedback,
 where we assume that the noise is independent of the system
 parameters ${\boldsymbol{\gamma }_i} = {\boldsymbol{\gamma }}$ for all $i$, with the
 functional model as,
\begin{align}\label{eqnMAC5}
Y_i  &= f_i \left( {X_{\left( 1 \right)i} ,...,X_{\left( K \right)i} ,W_i ,{\boldsymbol{\gamma }}} \right)\notag \\
 &= \sum\limits_{k = 1}^K {\gamma _{(k)} X_{\left( k \right)i} }  + W_i, \,\,\,\;i \in \left[ {1:n} \right]
\end{align}
and
\begin{eqnarray}\label{eqnMAC6}
X_{\left( k \right)i}  = g_{\left( k \right)i} \left( {U_{\left( k \right)}^i ,Y^{i - 1} } \right),\;k = 1,...,K,\,\;i \in \left[ {1:n} \right]
\end{eqnarray}
derivatives of MI with respect to $\gamma _{(l)} ,\;l = 1,...,K$ are,
\begin{eqnarray}\label{eqnforget}
\begin{array}{l}
 \frac{{\partial }}{{\partial \gamma_{(l)}}} I\left( {X_{\left( 1 \right)}^n ,...,X_{\left( K \right)}^n ;Y^n } \right)\\
 = \sum\limits_{k = 1}^K {{\mathop{\rm E}\nolimits} \left\{ {\left( \frac{{\partial {\bf{g}}_{\left( k \right)} }}{{\partial \gamma_{(l)} }}\right)^T\nabla _{X_{\left(k \right)}^n } \ln p_{Y^n \left| {X_{\left( 1 \right)}^n ,...,X_{\left( K \right)}^n } \right.} \left( {Y^n \left| {X_{\left( 1 \right)}^n ,...,X_{\left( K \right)}^n } \right.} \right)} \right\}}\\
\,\,\,\,\,+ {\mathop{\rm E}\nolimits} \left\{\left( {\sum\limits_{k = 1}^K {\frac{{\partial \gamma _{(k)} }}{{\partial \gamma _{(l)} }}
{\bf{g}}_{\left( k \right)} }  + \gamma_{(k)} \frac{{\partial {\bf{g}}_{\left( k \right)} }}{{\partial \gamma_{(l)} }}} \right)^T {\nabla _{Y^n } \ln p_{X_{\left( 1 \right)}^n ,...,X_{\left( K \right)}^n \left| {Y^n } \right.} \left( {X_{\left( 1 \right)}^n ,...,X_{\left( K \right)}^n \left| {Y^n } \right.} \right)} \right\}  \\
\end{array}
\end{eqnarray}
and
\begin{eqnarray}\label{eqnMAC7}
\begin{array}{l}
 \frac{{\partial }}{{\partial \gamma_{(l)}}} I\left( {U_{\left( 1 \right)}^n ,...,U_{\left( K \right)}^n ;Y^n } \right)\\
 = {\mathop{\rm E}\nolimits} \left\{\left( {\sum\limits_{k = 1}^K {\frac{{\partial \gamma _{(k)} }}{{\partial \gamma _{(l)} }}
{\bf{g}}_{\left( k \right)} }  + \gamma_{(k)} \frac{{\partial {\bf{g}}_{\left( k \right)} }}{{\partial \gamma_{(l)} }}} \right)^T {\nabla _{Y^n } \ln p_{U_{\left( 1 \right)}^n ,...,U_{\left( K \right)}^n \left| {Y^n } \right.} \left( {U_{\left( 1 \right)}^n ,...,U_{\left( K \right)}^n \left| {Y^n } \right.} \right)} \right\}  \\
 = \sum\limits_{i = 1}^n {{\mathop{\rm E}\nolimits} \left\{ \left( {\sum\limits_{k = 1}^K {\frac{{\partial \gamma _{(k)} }}{{\partial \gamma _{(l)} }}g_{_{\left( k \right)i} } }  + \gamma_{(k)} \frac{{\partial g_{_{\left( k \right)i} } }}{{\partial \gamma_{(l)} }}} \right){\frac{{\partial \ln p_{U_{\left( 1 \right)}^n ,...,U_{\left( K \right)}^n \left| {Y^n } \right.} \left( {U_{\left( 1 \right)}^n ,...,U_{\left( K \right)}^n \left| {Y^n } \right.}  \right)}}{{\partial Y_i }}} \right\}}
\end{array}
\end{eqnarray}
\end{col}
\begin{proof}
Proofs of~\eqref{eqnforget} and~\eqref{eqnMAC7} directly follow by substituting~\eqref{eqnMAC5} and~\eqref{eqnMAC6} in~\eqref{eqnMAC3} and~\eqref{eqnMAC4}, respectively and definition of the gradient of a scalar function with respect to a vector.
\end{proof}

In the following Corollary, we particularize this result for the Gaussian case.
\begin{col}(\textit{Gaussian Channel})\label{Corollary MAC2}
Consider the system model~\eqref{eqnMAC5} and~\eqref{eqnMAC6}, and let $W_i  \sim \mathcal{N}\left( {0,1} \right),\,i \in \left[ {1:n} \right]$ to be $i.i.d.$ samples of noise. Then,
we will have,
\begin{eqnarray}\label{eqnMAC88}
\frac{\partial }{{\partial \gamma_{(l)} }}I\left( {U_{\left( 1 \right)}^n ,...,U_{\left( K \right)}^n ;Y^n } \right)\; =\gamma_{(l)} mmse_{(l)}\left( \boldsymbol \gamma  \right) + \vartheta_{(l)} \left( \boldsymbol \gamma  \right)+ \zeta_{(l)} \left( \boldsymbol \gamma  \right)
\end{eqnarray}
where,
\begin{eqnarray}\label{eqnMAC99}
mmse_{\left( l \right)} \left( \boldsymbol \gamma  \right) =  \sum\limits_{i = 1}^n {{\mathop{\rm E}\nolimits} \left\{ {\left( {g_{_{\left( l \right)i} }  - {\mathop{\rm E}\nolimits} \left\{ {g_{_{\left( l \right)i} } \left| {Y^n } \right.} \right\}} \right)^2 } \right\}}
 \end{eqnarray}
and
\begin{eqnarray}\label{eqnMAC100}
\vartheta_{\left( l \right)}\left( \boldsymbol \gamma  \right) =  \sum\limits_{k = 1,k\ne l}^K {\gamma _{\left( k \right)}\left( {\sum\limits_{i = 1}^n {{\mathop{\rm E}\nolimits} \left\{ {\left( {g_{_{\left( l \right)i} }  - {\mathop{\rm E}\nolimits} \left\{ {g_{_{\left( l \right)i} } \left| {Y^n } \right.} \right\}} \right)\left( {g_{_{\left( k \right)i} }  - {\mathop{\rm E}\nolimits} \left\{ {g_{_{\left( k \right)i} } \left| {Y^n } \right.} \right\}} \right)} \right\}} } \right)}
\end{eqnarray}
and
\begin{eqnarray}\label{eqnMAC111}
\zeta_{\left( l \right)}\left( \boldsymbol \gamma  \right) = \sum\limits_{i = 1}^n {{\mathop{\rm E}\nolimits} \left\{ {\left( {\sum\limits_{k = 1}^K {\gamma _{\left( k \right)} \left( {g_{_{\left( k \right)i} }  - {\mathop{\rm E}\nolimits} \left\{ {g_{_{\left( k \right)i} } \left| {Y^n } \right.} \right\}} \right)} } \right)\left( {\sum\limits_{k = 1}^K {\gamma _{\left( k \right)} \frac{{\partial g_{_{\left( k \right)i} } }}{{\partial \gamma _{\left( l \right)} }}} } \right)} \right\}}
\end{eqnarray}
\end{col}
\begin{proof}
See Appendix~\ref{Gaussian MAC}.
\end{proof}

Now, consider the case that the system parameters of the different users are all the same, which may happen in a system with perfect power control. Unlike Corollary~\ref{Corollary MAC2}, in this case taking the derivative of MI with respect to each system parameter does not omit the other terms. Hence, this case can not be regarded as a special version of the above Corollary and the result for this case is given in the following Corollary.
\begin{col}(\textit{Gaussian Channel: Perfect power control})\label{Corollary MAC2_Prime}
Consider the setup of Corollary~\ref{Corollary MAC2} with $\gamma _{(l)}  = \gamma$ $\left(l = 1,...,K\right)$, then
\begin{eqnarray}\label{eqnMAC8}
\frac{\partial }{{\partial \gamma }}I\left( {U_{\left( 1 \right)}^n ,...,U_{\left( K \right)}^n ;Y^n } \right)\; = \gamma mmse\left( \gamma  \right) + \vartheta \left( \gamma  \right)+ \zeta \left( \gamma  \right)
\end{eqnarray}
where,
\begin{align}\label{eqnMAC9}
 mmse\left( \gamma  \right) &= \sum\limits_{k = 1}^K {mmse_{\left( k \right)} \left( \gamma  \right)}  \notag \\
 & =  \sum\limits_{k = 1}^K {\left( {\sum\limits_{i = 1}^n {{\mathop{\rm E}\nolimits} \left\{ {\left( {g_{_{\left( k \right)i} }  - {\mathop{\rm E}\nolimits} \left\{ {g_{_{\left( k \right)i} } \left| {Y^n } \right.} \right\}} \right)^2 } \right\}} } \right)}
 \end{align}
and
\begin{eqnarray}\label{eqnMAC10}
\vartheta\left( \gamma  \right) = \gamma \sum\limits_{k = 1}^K {\sum\limits_{\mathop {k' = 1}\limits_{k' \ne k} }^K {\left( {\sum\limits_{i = 1}^n {{\mathop{\rm E}\nolimits} \left\{ {\left( {g_{_{\left( k \right)i} }  - {\mathop{\rm E}\nolimits} \left\{ {g_{_{\left( k \right)i} } \left| {Y^n } \right.} \right\}} \right)\left( {g_{_{\left( {k'} \right)i} }  - {\mathop{\rm E}\nolimits} \left\{ {g_{_{\left( {k'} \right)i} } \left| {Y^n } \right.} \right\}} \right)} \right\}} } \right)} }
\end{eqnarray}
and
\begin{eqnarray}\label{eqnMAC11}
\zeta \left( \gamma  \right) = \gamma ^2 \sum\limits_{i = 1}^n {{\mathop{\rm E}\nolimits} \left\{ {\left( {\sum\limits_{k = 1}^K {\left( {g_{_{\left( k \right)i} }  - {\mathop{\rm E}\nolimits} \left\{ {g_{_{\left( k \right)i} } \left| {Y^n } \right.} \right\}} \right)} } \right)\left( {\sum\limits_{k = 1}^K {\frac{{\partial g_{_{\left( k \right)i} } }}{{\partial \gamma }}} } \right)} \right\}}
\end{eqnarray}
\end{col}
\begin{proof}
The proof of this Corollary is quite straightforward and similar to the proof of Corollary~\ref{Corollary MAC2}, and hence it is omitted.
\end{proof}

It is worth noting that, the first term of~\eqref{eqnMAC8} is sum of the MMSEs of the users,
the second term is caused by the interference of users and appears as a cross correlation
of the users' input estimation errors, and the last term is induced by feedback.
\subsubsection{Single-User Channel with Feedback}\label{single-user-with Feedback}
In this section, we reduce our general results on the MAC with feedback to the single-user communication channel with feedback, and recover the results of~\cite{EI-MMSE} on the single-user Gaussian channel with feedback.

In $n$ uses of a single-user channel, sequence of messages is denoted by $U^n$, and $X^n$ and $Y^n$ denote the sequences of channel input and channel output, respectively. We use $U_i$, $X_i$ and $Y_i$ to denote the message, channel input and channel output at time instant $i$. $Y^{i-1}$ is used for the channel outputs up to time instant $(i-1)$ and sequence of the messages up to time instant $i$ is denoted by $U^i$. It is assumed that at each time instant $i$ the encoder assigns a symbol to $U^i$ and previous received output sequence $Y^{i - 1}$.

Now, we give a general expression for the gradient of mutual
information with respect to the channel parameters in $n$ uses of a single-user channel, where in each use of the channel we assume that
the channel model is,
\begin{eqnarray}\label{eqn311}
Y_i  = f_i \left( {X_i ,W_i ,\gamma_i } \right)  ,\,\,\,\,\,\,i = 1,...,n
\end{eqnarray}
where the channel input $X_i$ depends on the sequence of messages $U^i$ and the previous channel outputs $Y^{i - 1}$ through the known
encoding function $g_i$ as,
\begin{eqnarray}\label{eqncoder}
 X_i  = g_i \left( {U^i ,Y^{i - 1} } \right).
 \end{eqnarray}
For this channel the following Corollary holds.
\begin{col}\label{col feedback}
Consider a single-user channel with channel model~\eqref{eqn311} for each use of the channel and~\eqref{eqncoder} as the coded channel input.
Then, derivatives of MI with
respect to $\gamma _j ,\;j = 1,...,n$ can be written as,
\begin{align}\label{eqn313}
\frac{\partial }{{\partial {\gamma _j}}}I\left( {{X^n};{Y^n}} \right) &= \sum\limits_{i = 1}^n {{\mathop{\rm E}\nolimits} \left\{ {\frac{{\partial {g_i}}}{{\partial {\gamma _j}}}\frac{\partial }{{\partial {X_i}}}\ln {p_{{Y^n}|{X^n}}}\left( {{Y^n}|{X^n}} \right)} \right\}} \notag \\
&\,\,\,\,\,\, + {\mathop{\rm E}\nolimits} \left\{ {\frac{{\partial {f_i}}}{{\partial {\gamma _j}}}\frac{\partial }{{\partial {Y_i}}}\ln {p_{{X^n}|{Y^n}}}\left( {{X^n}|{Y^n}} \right)} \right\}
 \end{align}
and
\begin{eqnarray}\label{eqn312}
\frac{\partial }{{\partial \gamma _j }}I\left( {U^n ;Y^n } \right) = \sum\limits_{i = 1}^n {{\mathop{\rm E}\nolimits} \left\{ {\frac{{\partial f_i }}{{\partial \gamma _j }}\frac{\partial }{{\partial Y_i }}\ln p_{U^n \left| {Y^n } \right.} \left( {U^n \left| {Y^n } \right.} \right)} \right\}}.
\end{eqnarray}
\end{col}
\begin{proof}
 Proof easily follows from~\eqref{eqnMAC3} and~\eqref{eqnMAC4} with $K=1$.
\end{proof}

Now, we particularize our general results for the derivative of MI in an additive noise channel with feedback.
\begin{col}\label{Corollary 1}
Consider the following system model for a single-user additive noise channel with feedback,
\begin{align}\label{eqn320}
Y_i  &= f_i \left( {X_i ,W_i ,\gamma } \right) \notag\\
&= \gamma g_i \left( {U^i ,Y^{i - 1} } \right) + W_i ,\,\,\,\,\,\,i = 1,...,n
\end{align}
where the channel input $X_i$ depends on $U^i$ and the previous channel outputs $Y^{i - 1}$ through the known encoding function $g_i$, $X_i  = g_i \left( {U^i ,Y^{i - 1} } \right)$, and we suppose that, $\gamma_i=\gamma$ for all $i$. Moreover, random variable $W_i$ stands for arbitrary but known-distribution noise.
Then,
\begin{eqnarray}\label{eqn321}
\begin{array}{l}
\frac{\partial }{{\partial \gamma }}I\left( {X^n ;Y^n } \right) \\
= \sum\limits_{i = 1}^n {\left( {{\mathop{\rm E}\nolimits} \left\{ {\frac{{\partial g_i }}{{\partial \gamma }}\frac{\partial }{{\partial X_i }}\ln p_{Y^n \left| {X^n } \right.} \left( {Y^n \left| {X^n } \right.}  \right)} \right\} + {\mathop{\rm E}\nolimits} \left\{ {\left( {g_i  + \gamma \frac{{\partial g_i }}{{\partial \gamma }}} \right)\frac{\partial }{{\partial Y_i }}\ln p_{X^n \left| {Y^n } \right.} \left( {X^n \left| {Y^n } \right.} \right)} \right\}} \right)}
\end{array}
\end{eqnarray}
and
\begin{eqnarray}\label{eqn2000}
\frac{\partial }{{\partial \gamma }}I\left( {U^n ;Y^n } \right) = \sum\limits_{i = 1}^n {{\mathop{\rm E}\nolimits} \left\{ {\left( {g_i  + \gamma \frac{{\partial g_i }}{{\partial \gamma }}} \right)\frac{\partial }{{\partial Y_i }}\ln p_{U^n \left| {Y^n } \right.} \left( {U^n \left| {Y^n } \right.}  \right)} \right\}}.
\end{eqnarray}
\end{col}
\begin{proof}
Proof readily follows by~\eqref{eqnforget} and~\eqref{eqnMAC7} with $K=1$, respectively.
\end{proof}

In the next Corollary we show that~\eqref{eqn2000} will reduce to the result of~\cite{EI-MMSE} for additive Gaussian noise channels.
\begin{col}(\textit{Gaussian Channel})\label{EX-IMMMSE}
Consider the single-user channel model~\eqref{eqn320} with $W_i  \sim \mathcal{N}\left( {0,1} \right),\,i \in \left[ {1:n} \right]$
to be $i.i.d.$ samples of additive Gaussian noise, then,
\begin{eqnarray}\label{eqn322}
\frac{\partial }{{\partial \gamma }}I\left( {U^n ;Y^n } \right) =
\gamma mmse\left( \gamma  \right) + \zeta \left( \gamma  \right).
\end{eqnarray}
where, $mmse\left( \gamma  \right)$ and $\zeta \left( \gamma  \right)$ are the minimum mean square estimation error of the channel input and the term induced by feedback defined by~\eqref{eqnMAC9} and~\eqref{eqnMAC11} with $K=1$.
\end{col}
\begin{proof}
Proof easily follows from Corollary~\ref{Corollary MAC2_Prime} with $K=1$.
\end{proof}
Therefore, we could recover the result
of~\cite{EI-MMSE} on Gaussian additive channel with feedback as a special case of our results.
\subsection{Channels without Feedback}\label{without Feedback}
As mentioned before, our general results are applicable for both
systems in which the system input is allowed to be a function of the system parameters or not. Here, we reduce
our results to the cases in which the channel input is not a function of the system parameters, where some of the available results in the literature are recovered.
\subsubsection{MAC without Feedback}\label{MAC-without Feedback}
 Consider a MAC without feedback, where based on Remarks~\ref{Remark 20} and~\ref{Remark 3} the reduced version of the previous section's results can be utilized for our analysis. 
\begin{col}\label{Corollary MAC-Without1}
Consider the system model~\eqref{eqnMAC1} and suppose that the channel inputs
are functions of the messages only. i.e.,
\begin{eqnarray}\label{eqnMACwithout1}
X_{\left( k \right)i}  = g_{\left( k \right)i} \left( {U_{\left( k \right)}^i  } \right);\;\;k = 1,...,K,\;i \in \left[ {1:n} \right]
\end{eqnarray}
 then,
\begin{align}\label{eqnMACwithout2}
\nabla _{{\boldsymbol{\gamma }}_j }I\left( {X_{\left( 1 \right)}^n ,...,X_{\left( K \right)}^n ;Y^n } \right)& = {\mathop{\rm E}\nolimits} \left\{ \left( \frac{{\partial {\bf{f}}}}{{\partial \boldsymbol\gamma_j }}\right)^T{\nabla _{Y^n } \ln p_{X_{\left( 1 \right)}^n ,...,X_{\left( K \right)}^n \left| {Y^n } \right.} \left( {X_{\left( 1 \right)}^n ,...,X_{\left( K \right)}^n \left| {Y^n } \right.} \right)} \right\}
 \end{align}
 and,
\begin{eqnarray}\label{eqnMACwithout3}
\nabla _{{\boldsymbol{\gamma }}_j } I\left( {U_{\left( 1 \right)}^n ,...,U_{\left( K \right)}^n ;Y^n } \right)= {\mathop{\rm E}\nolimits} \left\{ \left(\frac{{\partial {\bf{f}}}}{{\partial \boldsymbol\gamma_j }} \right)^T{\nabla _{Y^n } \ln p_{U_{\left( 1 \right)}^n ,...,U_{\left( K \right)}^n \left| {Y^n } \right.} \left( {U_{\left( 1 \right)}^n ,...,U_{\left( K \right)}^n \left| {Y^n } \right.} \right)} \right\}.
\end{eqnarray}
 \end{col}
 \begin{proof}
 Proofs easily follow from Remark~\ref{Remark 3}.
 \end{proof}

Following corollary specializes these results for the Gaussian MAC without feedback and recovers the scalar version of the result of~\cite{Ghanem}.
\begin{col}(\textit{Gaussian channel})\label{Corollary MAC-Without3}
Consider the system model~\eqref{eqnMAC5} with $\gamma _{(l)}  = \gamma$ for all $l$, and let $W_i  \sim \mathcal{N}\left( {0,1} \right),\,i \in \left[ {1:n} \right]$
to be $i.i.d.$ samples of noise. Then,
\begin{align}\label{eqnMACwithout5}
\frac{\partial }{{\partial \gamma }}I\left( {U_{\left( 1 \right)}^n ,...,U_{\left( K \right)}^n ;Y^n } \right)\; =\gamma mmse\left( \gamma  \right) + \vartheta \left( \gamma  \right)
\end{align}
where, $mmse\left( \gamma  \right)$ and $\vartheta \left( \gamma  \right)$ defined in~\eqref{eqnMAC9} and~\eqref{eqnMAC10}, respectively.
\end{col}
\begin{proof}
Substituting~\eqref{eqnMAC5} with Gaussian noise in~\eqref{eqnMACwithout3} gives the desired result.
\end{proof}

As it can be seen from~\eqref{eqnMACwithout5}, when there is no feedback in the channel model, derivative of MI will be related only to the sum of the MMSE terms of user's and cross correlation terms induced by the interference of different users.

\subsubsection{Single-User Channel without Feedback}\label{single-user-without Feedback}
Here we reduce our general results in the previous section to a point to point Gaussian channel, where the I-MMSE formula is recovered.

\begin{col}(\textit{Gaussian Channel})\label{Corollary Single-Without-Without}
Suppose that $W_i  \sim \mathcal{N}\left( {0,1} \right),\,i \in \left[ {1:n} \right]$
are $i.i.d.$ samples of additive Gaussian noise, then,
\begin{eqnarray}\label{eqnfeedfree6}
\frac{\partial }{{\partial \gamma }}I\left( {U^n ;Y^n } \right) =\gamma mmse\left( \gamma  \right)
\end{eqnarray}
\end{col}
\begin{proof}
Proof readily follow from~\eqref{eqnMACwithout5} with $K=1$.
\end{proof}
It is worth noting that~\eqref{eqnfeedfree6} is a general form of the well-known I-MMSE relationship reported in~\cite{I-MMSE}. Note that if $\gamma$ is considered as the square root of Signal-to-Noise Ratio (SNR), then the I-MMSE formula in~\cite{I-MMSE} will be recovered exactly.

\section{K-user Gaussian MAC with Feedback}\label{K-user Gaussian MAC Kramer}
In~\cite{Kramer2018} it is shown that, the sum-capacity of a K-user Gaussian MAC under symmetric power constraint $P$ is achievable using the Kramer's F-MEC code proposed in~\cite{Kramer2002}. Here, we will show that how this coding strategy affects the information rate of a K-user Gaussian MAC.   

\subsection{System Model and Coding Scheme }\label{model_K_User}
In the $i$-th use of a K-user Gaussian MAC, the received signal is,
\begin{eqnarray}\label{eqnKera_1}
{Y_i} = \sum\limits_{k = 1}^K {{h_k}{X_{\left( k \right)i}}}  + {W_i};\,\,\,\,\,\,\,\,i = 1,2,...,n
\end{eqnarray}
where ${h_k},\,\, k=1,2,...,K$ are the channel gains, and $W_i  \sim \mathcal{N}\left( {0,1} \right)\,\,\left( i \in \left[ {1:n} \right]\right)$ stands for the $i.i.d.$ samples of the Gaussian noise. The channel inputs have the power constraints
\begin{eqnarray}\label{eqnKera_2}
\sum\limits_{i = 1}^n {{\mathop{\rm E}\nolimits} \left\{ {X_{\left( k \right)i}^2} \right\}}  \le n{P_k},\,\,\,\,\,\,\,k = 1,2,...,K
\end{eqnarray}

The transmitted signal in the $i$-th use of the channel is
\begin{eqnarray}\label{eqnKera_3}
{X_{\left( k \right)i}} = {g_{\left( k \right)i}}\left( {{M_{\left( k \right)}},{Y^{i - 1}}} \right),\,\,\,\,\,\,k = 1,2,...,K
\end{eqnarray}
where the ${g_{\left( k \right)i}}\left(  \cdot  \right)$ are the encoding functions. Following the notation of~\cite{Kramer2002}, the encoder maps the
message ${M_{\left( k \right)}}$ onto a point ${\theta _{\left( k \right)}}$ in the complex
plane, and ${\hat \theta _{\left( k \right)i}}$ is the receiver's Linear Minimum Mean Squared Error (LMMSE) estimate of ${\theta _{\left( k \right)}}$ after $i$ channel uses. The transmitters send power-normalized versions of the estimation error, i.e., ${\varepsilon _{\left( k \right)i-1}} = {\hat \theta _{\left( k \right)i-1}} - {\theta _{\left( k \right)}}$ as
\begin{eqnarray}\label{eqnKera_4}
{X_{\left( k \right)i}} = \left( {\sqrt {{P_k \mathord{\left/
 {\vphantom {P {\sigma _{\left( k \right)i - 1}^2}}} \right.
 \kern-\nulldelimiterspace} {\sigma _{\left( k \right)i - 1}^2}}} m_{\left( k \right)i}^*} \right){\varepsilon _{\left( k \right)i - 1}},\,\,\,\,\,k=1,2,...,K
\end{eqnarray}
where ${\gamma _{\left( k \right)i}} \buildrel \Delta \over = \sqrt {{P_k \mathord{\left/
 {\vphantom {P {\sigma _{\left( k \right)i - 1}^2}}} \right.
 \kern-\nulldelimiterspace} {\sigma _{\left( k \right)i - 1}^2}}} m_{\left( k \right)i}^*$ is a power scaling factor, $\sigma _{\left( k \right)i - 1}^2 = {\mathop{\rm E}\nolimits} \left\{ {{{\left| {{\varepsilon _{\left( k \right)i - 1}}} \right|}^2}} \right\}$, and ${m_{\left( k \right)i}}$ is a modulation coefficient chosen as
\begin{eqnarray}\label{eqnKera_6}
{m_{\left( k \right)i}} = \exp \left( {j\frac{{2\pi \left( {k - 1} \right)}}{K}\left( {i - 1} \right)} \right),\,\,\,\,\,i = 1,2,...,n.
\end{eqnarray}
The $K^2$ user's correlation coefficients can be written as,
\begin{eqnarray}\label{eqnKera_9}
{\rho _{\left( {kk'} \right)i}} = \frac{{{\mathop{\rm E}\nolimits} \left\{ {{\varepsilon _{\left( k \right)i}}\varepsilon _{\left( {k'} \right)i}^*} \right\}}}{{\sqrt {\sigma _{\left( k \right)i}^2\sigma _{\left( {k'} \right)i}^2} }},\,\,\,1 \le k \le K,\,1 \le k' \le K.
\end{eqnarray}

In a symmetric Gaussian MAC i.e., ${P_k} = P$ and $h_k=1$ for all $k$, the sum-capacity 
\begin{eqnarray}\label{eqnKera_C_SUM_1}
{C_{sum}} = \frac{1}{2}\ln \left( {1 + PK\beta } \right)
\end{eqnarray}
is achievable using F-MEC coding strategy~\cite{Kramer2018}, if $P$ satisfies 
\begin{eqnarray}\label{eqnKera_C_SUM_3}
{\left( {{{P{K^2}} \mathord{\left/
 {\vphantom {{P{K^2}} 2}} \right.
 \kern-\nulldelimiterspace} 2} + 1} \right)^{K - 1}} \le {\left( {{{P{K^2}} \mathord{\left/
 {\vphantom {{P{K^2}} 4}} \right.
 \kern-\nulldelimiterspace} 4} + 1} \right)^K}
\end{eqnarray}
and $\beta  \in \left[ {1,K} \right]$ is the unique solution of
\begin{eqnarray}\label{eqnKera_C_SUM_2}
{\left( {1 + PK\beta } \right)^{K - 1}} = {\left( {1 + P\beta \left( {K - \beta } \right)} \right)^K}.
\end{eqnarray}

\subsection{Feedback Effect Analysis of the K-user Symmetric Gaussian MAC}\label{Feedback analysis_K-User}
In the following, utilizing the derivative of the MI we show that how F-MEC code affects the information rate of a K-user symmetric Gaussian MAC. Toward this end, the scaling factor ${\gamma _{\left( k \right)i}}$ is decomposed into two multiplicative factors as, ${\gamma _{\left( k \right)i}} = \sqrt[4]{P}\sqrt {{{{P^{{1 \mathord{\left/
 {\vphantom {1 2}} \right.
 \kern-\nulldelimiterspace} 2}}}} \mathord{\left/
 {\vphantom {{{P^{{1 \mathord{\left/
 {\vphantom {1 2}} \right.
 \kern-\nulldelimiterspace} 2}}}} {\sigma _{\left( k \right)i - 1}^2}}} \right.
 \kern-\nulldelimiterspace} {\sigma _{\left( k \right)i - 1}^2}}} m_{\left( k \right)i}^*$, and hence,
\begin{eqnarray}\label{eqnKera_7}
{Y_i} = \sum\limits_{k = 1}^K {\sqrt[4]{P} {{\tilde X}_{\left( k \right)i}}}  + {W_i};\,\,\,\,\,\,\,\,i = 1,2,...,n
\end{eqnarray}
where ${\tilde X_{\left( k \right)i}} =\sqrt {{{{P^{{1 \mathord{\left/
 {\vphantom {1 2}} \right.
 \kern-\nulldelimiterspace} 2}}}} \mathord{\left/
 {\vphantom {{{P^{{1 \mathord{\left/
 {\vphantom {1 2}} \right.
 \kern-\nulldelimiterspace} 2}}}} {\sigma _{\left( k \right)i - 1}^2}}} \right.
 \kern-\nulldelimiterspace} {\sigma _{\left( k \right)i - 1}^2}}} m_{\left( k \right)i}^* {\varepsilon _{\left( k \right)i - 1}}$. It should be noted that, considering the system parameter ${\gamma _{\left( k \right)i}}$ as the above does not change the transmitted power or coding scheme, but it allows us to study the effects of noise, interference and feedback, distinctively.

\begin{lemma}
Considering the system model~\eqref{eqnKera_7} and applying the Kramer's F-MEC coding strategy, the derivative of the normalized MI to the number of channel uses, i.e., $n$ with respect to the SNR of each user $P$, will be as 
\begin{eqnarray}\label{eqnKera_8}
\frac{\partial }{{\partial P}}I_{F-MEC}\left( {{M_{\left( 1 \right)}},{M_{\left( 2 \right)}},...,{M_{\left( K \right)}};{Y^{i - 1}}} \right) = mmse_{F-MEC}(P) + {\vartheta}_{F-MEC}(P)  + {\zeta}_{F-MEC}(P) 
\end{eqnarray}
where
\begin{eqnarray}\label{eqnKera_9}
mmse_{F-MEC}(P) = \frac{K}{{4P}}\left( {P - 2\kappa P\left( {1 + \rho \left( {K - 1} \right)} \right) + {\kappa ^2}\sigma _Y^2} \right)
\end{eqnarray}
and
\begin{eqnarray}\label{eqnKera_10}
{\vartheta}_{F-MEC}(P)  = \frac{{K\left( {K - 1} \right)}}{{4P}}\left( {P\rho  - 2\kappa P\left( {1 + \left( {K - 1} \right)\rho } \right) + {\kappa ^2}\sigma _Y^2} \right)
\end{eqnarray}
and
\begin{eqnarray}\label{eqnKera_11}
{\zeta}_{F-MEC}(P)  = \frac{1}{4}K\left( {1 - K\kappa } \right)\left( {1 + \rho \left( {K - 1} \right)} \right)
\end{eqnarray}
where $\kappa  = {{P\left( {1 + \left( {K - 1} \right)\rho } \right)} \mathord{\left/
 {\vphantom {{P\left( {1 + \left( {K - 1} \right)\rho } \right)} {\sigma _Y^2}}} \right.
 \kern-\nulldelimiterspace} {\sigma _Y^2}}$, $\sigma _Y^2 = KP\left( {1 + \left( {K - 1} \right)\rho } \right) + 1$ is the received signal variance, and $\rho$ is obtained from $\beta  = 1 + \left( {K - 1} \right)\rho $.
\end{lemma}
\begin{proof}
This lemma can be proved using the same steps
as in the proof of Corollary~\ref{Corollary MAC2_Prime} and considering the system model~\eqref{eqnKera_7} with the introduced Kramer's coding scheme. 
\end{proof}

\begin{remark}\label{Remark_Kramer_1}
Substituting $\kappa $ in~\eqref{eqnKera_9}{--}~\eqref{eqnKera_11}, surprisingly we find that,
\begin{align}\label{eqnS_K1011}
{\zeta}_{F-MEC}(P)  = {mmse}_{F-MEC}(P) + {\vartheta }_{F-MEC}(P).
\end{align}
Moreover, replacing $\kappa $ in~\eqref{eqnKera_11} shows that the term induced by feedback, i.e., $\zeta_{F-MEC}$ is always positive and hence this term increases the information rate versus the transmitted power. These observations can be explained by the fact that, the capacity achieving F-MEC code uses the information of both noise and users' interference contained in the noiseless feedback such that the feedback term alone compensates the effects of the other terms. Fig.~\ref{fig:MEC_8} shows the effects of the $mmse_{F-MEC}$, ${\vartheta}_{F-MEC}$ and ${\zeta}_{F-MEC}$ in a 8-user symmetric Gaussian MAC with F-MEC coding strategy.

\begin{figure}
  \centering
  \includegraphics[width=10cm]{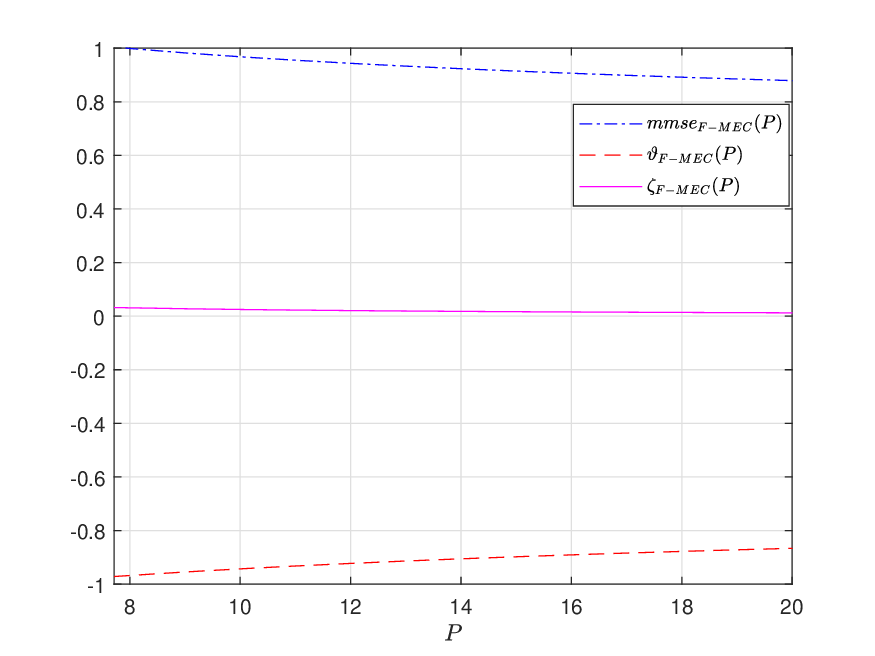}
 \caption{Derivative of the MI with respect to the transmitted power of each user $P$, in a 8-user symmetric Gaussian MAC with F-MEC coding scheme.}
 \label{fig:MEC_8}
\end{figure}
\end{remark}

\begin{remark}\label{Remark_Kramer_2}
For the point to point additive white Gaussian noise channel, i.e., $K=1$, the interference term of~\eqref{eqnKera_8} will be omitted and hence we can write,
\begin{eqnarray}\label{eqnKera_Rem}
\frac{\partial }{{\partial P}}I_{F-MEC}\left( {{M_{\left( 1 \right)}};{Y^{i - 1}}} \right) = mmse_{F-MEC}(P) + {\zeta}_{F-MEC}(P)  = \frac{1}{{2\left( {1 + P} \right)}}.
\end{eqnarray}
Integrating both sides of this equation over $P$ gives $I_{F-MEC}\left( {{M_{\left( 1 \right)}};{Y^{i - 1}}} \right) = \frac{1}{2}\ln \left( {1 + P} \right)$ which verifies the fact that
feedback does not increase the capacity of a white Gaussian
channel. Moreover, one can reduce the general result in~\eqref{eqnKera_8} for the case of $K=2$, where integrating the result over $P$ gives the sum-capacity of the two-user symmetric Gaussian MAC which can be achieved by the Schalkwijk{--}Kailath coding scheme~\cite{Network}.
\end{remark}

\section{Conclusions}\label{Conclusions}
We introduced a general expression for the gradient of MI in a general system model without imposing any constraints regarding to the relationship between the system elements. We shed light on the new aspect of the importance of the gradient of MI by using it to interpret the effect of feedback in MAC as a realistic scenario. In the additive Gaussian noise MAC with feedback, the obtained extension of I-MMSE formula relates the derivative of MI to the MMSE term resulted from the estimating of the channel inputs based on the channel output, the cross correlation of the estimation errors of the inputs, and an additional term caused by feedback. Considering F-MEC coding strategy, our obtained results on the gradient of MI clearly specified the constructive role of feedback on increasing the achievable rates of the users in a K-user symmetric Gaussian MAC. Also, the results were specialized for the single-user channels with and without feedback, where several available results were recovered.

\appendix
\subsection{Proof of Theorem~\ref{Theorem General} }\label{proof Theorem General}
  Using the relationship between MI and entropy, we can write,
\begin{align}\label{eqn3}
 \Delta I \left( {{\textbf{x}};{\textbf{y}}} \right)&= [{h\left( \textbf x +{\boldsymbol{\delta }}_{\textbf x} \right) - h\left( {\textbf{x}} \right)}] \notag\\
 &\,\,\,\,\,+ [{h\left( \textbf y +{\boldsymbol{\delta }}_{\textbf y}  \right) - h\left( {\textbf{y}} \right)}] \notag\\
 &\,\,\,\,\,- [{h\left( \textbf x +{\boldsymbol{\delta }}_{\textbf x}  ,\textbf y +{\boldsymbol{\delta }}_{\textbf y}  \right) - h\left( {{\textbf{x}},{\textbf{y}}} \right)}]
 \end{align}
where, $h(\cdot)$ denotes the  Shannon differential entropy,
and ${\boldsymbol{\delta }}_{\textbf x} $ and ${\boldsymbol{\delta }}_{\textbf y} $
denote the variations of $\textbf x$ and $\textbf y$ given by~\eqref{eqnDelta_x} and~\eqref{eqnDelta_y}, respectively.

Now, we calculate the first term of~\eqref{eqn3} where the other terms can be found similarly. From~\cite{Pham} and~\cite{BabaiezadehThesis}, we have,
\begin{align}\label{eqn22228}
h\left( \textbf x +{\boldsymbol{\delta }}_{\textbf x} \right) - h\left( {\textbf{x}} \right)&=
{\mathop{\rm E}\nolimits} \left\{ {\ln p_{\textbf{x}} \left( {{\textbf{x}}} \right) - \ln p_{\textbf{x}} \left( {{{\textbf{x}}} + {{\boldsymbol{\delta }}_{{\textbf{x}}}} }  \right)} \right\}\notag \\
&\,\,\,\,\,- {\mathop{\rm E}\nolimits} \left\{ {\ln \left( {{{p_{\textbf x +{\boldsymbol{\delta }}_{\textbf{x}} } \left( {{{\textbf{x}}} + {{\boldsymbol{\delta }}_{{\textbf{x}}}} }  \right)} \mathord{\left/
 {\vphantom {{p_{\textbf x +{\boldsymbol{\delta }}_{\textbf x}  } \left( {{{\textbf{x}}} + {{\boldsymbol{\delta }}_{{\textbf{x}}}} } \right)} {p_{\textbf{x}} \left( {{{\textbf{x}}} + {{\boldsymbol{\delta }}_{{\textbf{x}}}} }  \right)}}} \right.
 \kern-\nulldelimiterspace} {p_{\textbf{x}} \left( {{{\textbf{x}}} + {{\boldsymbol{\delta }}_{{\textbf{x}}}} }  \right)}}} \right)} \right\}
\end{align}
where, ${p_{{\textbf{x}} + {\boldsymbol{\delta }}_{\textbf{x}} } \left( {{{\textbf{x}}} + {{\boldsymbol{\delta }}_{{\textbf{x}}}} } \right)}$ is the
joint pdf of the random vector ${{\textbf{x}} + {\boldsymbol{\delta }}_{\textbf{x}} }$.
By Taylor series expansion of $\ln(\cdot)$ in the neighborhood of $1$ and using~\cite[Lemma 4]{Pham}, we can write,
\begin{eqnarray}\label{eqn22225}
{\mathop{\rm E}\nolimits} \left\{ {{{\left[ {\ln \frac{{p_{\textbf{x}} \left( {{{\textbf{x}}} + {{\boldsymbol{\delta }}_{{\textbf{x}}}} } \right)}}{{p_{{\textbf{x}} + {\boldsymbol{\delta }}_{\textbf{x}} } \left( {{{\textbf{x}}} + {{\boldsymbol{\delta }}_{{\textbf{x}}}} }\right)}} - \frac{{p_{\textbf{x}} \left( {{{\textbf{x}}} + {{\boldsymbol{\delta }}_{{\textbf{x}}}} } \right)}}{{p_{{\textbf{x}} + {\boldsymbol{\delta }}_{\textbf{x}} } \left( {{{\textbf{x}}} + {{\boldsymbol{\delta }}_{{\textbf{x}}}} } \right)}} + 1} \right]} \mathord{\left/
 {\vphantom {{\left[ {\ln \frac{{p_{\bf{x}} \left( {{\bf{x}} + {\bf{\delta }}_{\bf{\gamma }} } \right)}}{{p_{{\bf{x}} + {\bf{\delta }}_{\bf x} } \left( {{\bf{x}} + {\bf{\delta }}_{\bf x} } \right)}} - \frac{{p_{\bf{x}} \left( {{\bf{x}} + {\bf{\delta }}_{\bf x} } \right)}}{{p_{{\bf{x}} + {\bf{\delta }}_{\bf x} } \left( {{\bf{x}} + {\bf{\delta }}_{\bf x} } \right)}} + 1} \right]} {\left\| \boldsymbol \varepsilon \right\|}}} \right.
 \kern-\nulldelimiterspace} {\left\| \boldsymbol \varepsilon \right\|}}} \right\}\mathop  \to \limits^{\boldsymbol \varepsilon  \to {\bf{0}}} 0.
\end{eqnarray}
Equivalently, we can write,
\begin{eqnarray}\label{eqn22226}
{\mathop{\rm E}\nolimits} \left\{ {\ln \frac{{p_{\textbf{x}} \left( {{{\textbf{x}}} + {{\boldsymbol{\delta }}_{{\textbf{x}}}} }\right)}}{{p_{{\textbf{x}} + {\boldsymbol{\delta }}_{\textbf{x}} } \left( {{{\textbf{x}}} + {{\boldsymbol{\delta }}_{{\textbf{x}}}} } \right)}}} \right\} = {\mathop{\rm E}\nolimits} \left\{ {\frac{{p_{\textbf{x}} \left( {{{\textbf{x}}} + {{\boldsymbol{\delta }}_{{\textbf{x}}}} } \right)}}{{p_{{\textbf{x}} + {\boldsymbol{\delta }}_{\textbf{x}} } \left( {{{\textbf{x}}} + {{\boldsymbol{\delta }}_{{\textbf{x}}}} }\right)}} - 1} \right\} + o\left( \boldsymbol \varepsilon \right).
\end{eqnarray}
Hence,
\begin{align}\label{eqn22227}
 {\mathop{\rm E}\nolimits} \left\{ {\ln \frac{{p_{\textbf{x}} \left( {{{\textbf{x}}} + {{\boldsymbol{\delta }}_{{\textbf{x}}}} }\right)}}{{p_{{\textbf{x}} + {\boldsymbol{\delta }}_{\textbf{x}} } \left( {{{\textbf{x}}} + {{\boldsymbol{\delta }}_{{\textbf{x}}}} } \right)}}} \right\}  &= \int {\left( {\frac{{p_{\textbf{x}} \left( \textit{\textbf t}\right)}}{{p_{{\textbf{x}} + {\boldsymbol{\delta }}_{\textbf{x }} } \left( \textit{\textbf t} \right)}} - 1} \right)p_{{\textbf{x}} + {\boldsymbol{\delta }}_{\textbf{x }} } \left( \textit{\textbf t}\right)d\left( \textit{\textbf t} \right)}  \notag \\
 &\,\,\,\,\, + o\left( \boldsymbol \varepsilon\right) \notag \\
 &= o\left( \boldsymbol \varepsilon\right).
\end{align}
Substituting~\eqref{eqn22227} in~\eqref{eqn22228} we can write,
\begin{eqnarray}\label{eqn22229}
h\left( \textbf x +{\boldsymbol{\delta }}_{\textbf x} \right) - h\left( {\textbf{x}} \right)&=
{\mathop{\rm E}\nolimits} \left\{ {\ln p_{\textbf{x}} \left( {{\textbf{x}}} \right) - \ln p_{\textbf{x}} \left( {{{\textbf{x}}} + {{\boldsymbol{\delta }}_{{\textbf{x}}}} }  \right)} \right\}+ o\left( \boldsymbol \varepsilon\right).
\end{eqnarray}
In what follows we show that
\begin{eqnarray}\label{eqn22230}
{\mathop{\rm E}\nolimits} \left\{ {{{\left[ {\ln p_{\textbf{x}} \left( {{\textbf{x}} } \right) - \ln p_{{\textbf{x}}  } \left( {{{\textbf{x}}} + {\boldsymbol{\delta }}_{\textbf{x}} } \right) + {\boldsymbol{\phi }}_{\textbf{x}}^T \left({\textbf{x}} \right) {\bf{\tilde X\boldsymbol \varepsilon }}} \right]} \mathord{\left/
 {\vphantom {{\left[ {\ln p_{\textbf{x}} \left( {\textbf{x}} \right) - \ln p_{{\textbf{x}}  } \left( {{\textbf{x}} + {\bf{\delta }}_{\bf{x}} } \right) - {\boldsymbol{\phi }}_{\textbf{x}}^T {\bf{\tilde X\boldsymbol \varepsilon }}} \right]} {\left\| \boldsymbol \varepsilon \right\|}}} \right.
 \kern-\nulldelimiterspace} {\left\| \boldsymbol \varepsilon \right\|}}} \right\}\mathop  \to \limits^{\boldsymbol \varepsilon  \to {\bf{0}}} 0.
\end{eqnarray}
Toward this end, we first show that,
\begin{eqnarray}\label{eqn22231}
{\mathop{\rm E}\nolimits} \left\{ {{{\left[ {\ln p_{\textbf{x}} \left( {{\textbf{x}} } \right) - \ln p_{{\textbf{x}} } \left( {{{\textbf{x}}} + {\boldsymbol{\delta }}_{\textbf{x}} } \right) + {\boldsymbol{\phi }}_{\textbf{x}}^T \left( {{\textbf{x}}} \right){\boldsymbol{\delta }}_{\textbf{x}}  } \right]} \mathord{\left/
 {\vphantom {{\left[ {\ln p_{\textbf{x}} \left( {\textbf{x}} \right) - \ln p_{{\textbf{x}} } \left( {{\textbf{x}} + {\bf{\delta }}_{\bf{x}} } \right) - {\boldsymbol{\phi }}_{\textbf{x}}^T \left( {\textbf{x}} \right) } \right]} {\left\| {\boldsymbol \varepsilon } \right\|}}} \right.
 \kern-\nulldelimiterspace} {\left\| {\boldsymbol \varepsilon} \right\|}}} \right\}\mathop  \to \limits^{\boldsymbol \varepsilon  \to {\bf{0}}} 0.
\end{eqnarray}
Each realization of the random variable inside the above curly brackets can be written as,
\begin{eqnarray}\label{eqnpro_detail1}
\begin{array}{l}
 \left( {{{\left( {\ln p_{\textbf{x}} \left( \textit{\textbf{x}} \right) - \ln p_{\textbf{x}} \left( {\textit{\textbf{x}}  + {\boldsymbol{\delta }}_{\textit{\textbf{x}} } } \right) + \boldsymbol \phi _{\textbf{x}}^{\rm{T}} \left( \textit{\textbf{x}} \right){\boldsymbol{\delta }}_{\textit{\textbf{x}}} } \right)} \mathord{\left/
 {\vphantom {{\left( {\ln p_{\textbf{x}} \left( \textit{\textbf{x}} \right) - \ln p_{\textbf{x}} \left( {\textit{\textbf{x}}+ {\boldsymbol{\delta }}_{\textit{\textbf{x}}} } \right) + \boldsymbol \phi _{\textbf{x}}^{\rm{T}} \left( \textit{\textbf{x}} \right){\boldsymbol{\delta }}_{\textit{\textbf{x}}} } \right)} {\left\| {{\boldsymbol{\delta }}_{\textit{\textbf{x}}} } \right\|}}} \right.
 \kern-\nulldelimiterspace} {\left\| {{\boldsymbol{\delta }}_{\textit{\textbf{x}}} } \right\|}}} \right) \times \left( {{{\left\| {{\boldsymbol{\delta }}_{\textit{\textbf{x}}} } \right\|} \mathord{\left/
 {\vphantom {{\left\| {{\boldsymbol{\delta }}_{\textit{\textbf{x}}} } \right\|} {\left\| {\boldsymbol{\varepsilon }} \right\|}}} \right.
 \kern-\nulldelimiterspace} {\left\| {\boldsymbol{\varepsilon }} \right\|}}} \right) \\
  = \left( {{{\left( {\ln p_{\textbf{x}} \left( \textit{\textbf{x}} \right) - \ln p_{\textbf{x}} \left( {\textit{\textbf{x}} + \boldsymbol{\tilde X}{\boldsymbol{\varepsilon }} + o \left( {\boldsymbol{\varepsilon }} \right)} \right) + \boldsymbol \phi _{\textbf{x}}^{\rm{T}}\left( \textit{\textbf{x}} \right) \left( {\boldsymbol{\tilde X}{\boldsymbol{\varepsilon }} + o \left( {\boldsymbol{\varepsilon }} \right)} \right)} \right)} \mathord{\left/
 {\vphantom {{\left( {\ln p_{\textbf{x}} \left( \textit{\textbf{x}} \right) - \ln p_{\textbf{x}} \left( {\textit{\textbf{x}} + \boldsymbol{\tilde X}{\boldsymbol{\varepsilon }} + o \left( {\boldsymbol{\varepsilon }} \right)} \right) + \boldsymbol \phi _{\textbf{x}}^{\rm{T}} \left( \textit{\textbf{x}} \right)\left( {\boldsymbol{\tilde X}{\boldsymbol{\varepsilon }} + o \left( {\boldsymbol{\varepsilon }} \right)} \right)} \right)} {\left\| {\boldsymbol{\tilde X}{\boldsymbol{\varepsilon }} + o \left( {\boldsymbol{\varepsilon }} \right)} \right\|}}} \right.
 \kern-\nulldelimiterspace} {\left\| {\boldsymbol{\tilde X}{\boldsymbol{\varepsilon }} + o \left( {\boldsymbol{\varepsilon }} \right)} \right\|}}} \right) \\
 \,\,\,\,\, \times \left( {\left\| {{{\boldsymbol{\tilde X}{\boldsymbol{\varepsilon }}} \mathord{\left/
 {\vphantom {{\boldsymbol{\tilde X}{\boldsymbol{\varepsilon }}} {\left\| {\boldsymbol{\varepsilon }} \right\| + {{o \left( {\boldsymbol{\varepsilon }} \right)} \mathord{\left/
 {\vphantom {{o \left( {\boldsymbol{\varepsilon }} \right)} {\left\| {\boldsymbol{\varepsilon }} \right\|}}} \right.
 \kern-\nulldelimiterspace} {\left\| {\boldsymbol{\varepsilon }} \right\|}}}}} \right.
 \kern-\nulldelimiterspace} {\left\| {\boldsymbol{\varepsilon }} \right\| + {{o \left( {\boldsymbol{\varepsilon }} \right)} \mathord{\left/
 {\vphantom {{o \left( {\boldsymbol{\varepsilon }} \right)} {\left\| {\boldsymbol{\varepsilon }} \right\|}}} \right.
 \kern-\nulldelimiterspace} {\left\| {\boldsymbol{\varepsilon }} \right\|}}}}} \right\|} \right) \\
 \end{array}
 \end{eqnarray}
By Taylor series expansion of $\ln p_{{\textbf{x}}} \left( {\textbf{{\textit{x}}}} \right)$ about $\textit{\textbf x}$,~\eqref{eqnpro_detail1} converges to $0$ almost
surely as ${\boldsymbol \varepsilon  \to {\bf{0}}}$. Now, by the Lebesgue dominated convergence Theorem~\cite{Rudin}, to prove~\eqref{eqn22231} it suffices only
 to show the random variable inside the curly bracket in~\eqref{eqn22231} for all $\boldsymbol \varepsilon $ small enough is bounded by a fixed integrable random variable. Repeating some arguments like those of~\cite[Proof of Lemma 2]{Pham}, we can write,
\begin{eqnarray}\label{eqn22232}
\begin{array}{l}
\left|{{{\left[ {\ln p_{\textbf{x}} \left( \textit{{\textbf{x}} } \right) - \ln p_{{\textbf{x}} } \left( {\textit{{\textbf{x}}} + {\boldsymbol{\delta }}_{\textit{\textbf{x}}} } \right) + {\boldsymbol{\phi }}_{\textbf{x}}^T \left( {\textit{\textbf{x}}} \right){\boldsymbol{\delta }}_{\textit{\textbf{x}}}  } \right]} \mathord{\left/
 {\vphantom {{\left[ {\ln p_{\textbf{x}} \left( {\textbf{x}} \right) - \ln p_{{\textbf{x}} + {\boldsymbol{\delta }}_{\textbf{x}} } \left( {{\textbf{x}} + {\bf{\delta }}_{\textbf{x}} } \right) - {\boldsymbol{\phi }}_{\textbf{x}}^T \left( {\textbf{x}} \right) } \right]} {\left\| {\boldsymbol \varepsilon } \right\|}}} \right.
 \kern-\nulldelimiterspace} {\left\| {\boldsymbol \varepsilon } \right\|}}} \right| \\
  \,\,\,\,\,\,\mathop  <  C\left[ {2 + 2^{\max \left( {\alpha  - 2,0} \right)} \left( {\left\| {\textit{\textbf{x}}} \right\|^{\alpha  - 1}  + \left\| {{{\boldsymbol{\tilde X} \boldsymbol\varepsilon }}+ o\left( \boldsymbol \varepsilon \right)} \right\|^{\alpha  - 1} } \right)} + \left\| {\textit{\textbf{x}}} \right\|^{\alpha  - 1}  \right]\left\| {{\boldsymbol{\tilde X}  }} \right\|  \\
 \,\,\,\,\,\, < C\left[ {2 + 2^{\max \left( {\alpha  - 2,0} \right)} \left( {\left\| {\textit{\textbf{x}}} \right\|^{\alpha  - 1}  + \left\| \boldsymbol \varepsilon \right\|^{\alpha  - 1} \left\| {{{\left( {{{\boldsymbol{\tilde X} \boldsymbol\varepsilon }}  + o\left( \boldsymbol \varepsilon \right)} \right)} \mathord{\left/
 {\vphantom {{\left( {{{\boldsymbol{\tilde X} \boldsymbol\varepsilon }}  + o\left( {{\boldsymbol{\delta }}_{\boldsymbol{\gamma }} } \right)} \right)} {\left\| {{{\boldsymbol{\delta }}_{\boldsymbol{\gamma }} } } \right\|}}} \right.
 \kern-\nulldelimiterspace} {\left\| \boldsymbol \varepsilon \right\|}}} \right\|^{\alpha  - 1}  } \right)} + \left\| {\textit{\textbf{x}}} \right\|^{\alpha  - 1} \right]\left\| {{\boldsymbol{\tilde X} }} \right\|  \\
 \,\,\,\,\,\, < C\left[ {2 + 2^{\max \left( {\alpha  - 2,0} \right)} \left( {\left\| {\textit{\textbf{x}}} \right\|^{\alpha  - 1}  + \beta\left\| {{{\boldsymbol{\tilde X}  }}} \right\|^{\alpha  - 1}  } \right)}+ \left\| {\textit{\textbf{x}}} \right\|^{\alpha  - 1}  \right]\left\| {{{\boldsymbol{\tilde X}  }}} \right\|
 \end{array}
\end{eqnarray}
where, the first inequality follows by our assumption on the score functions to be bounded, the mean value Theorem~\cite{Rudin} and the fact that $\left( {a + b} \right)^\lambda   \le 2^{\max \left( {\lambda  - 1,0} \right)} \left( {a^\lambda   + b^\lambda  } \right)
$ holds for positive $a$, $b$, and $\lambda$. The second inequality follows by the assumption of $\boldsymbol \varepsilon$ to be small enough, and the last inequality holds for all $\boldsymbol \varepsilon$ such that $\left\| \boldsymbol \varepsilon\right\| \le \beta
$. This upper bound is integrable by Holder inequality~\cite{Rudin}, and hence~\eqref{eqn22231} is verified.
 On the other hand, since $\mathop {\lim }\limits_{\boldsymbol \varepsilon  \to {\bf{0}}} {{o\left( \boldsymbol \varepsilon \right)} \mathord{\left/
 {\vphantom {{o\left( \boldsymbol \varepsilon \right)} {\left\| \boldsymbol \varepsilon \right\|}}} \right.
 \kern-\nulldelimiterspace} {\left\| \boldsymbol \varepsilon \right\|}} = 0 $ and we assume that the
 score functions are bounded,~\eqref{eqn22230} readily follows. From~\eqref{eqn22229} and~\eqref{eqn22230} we get the desired result,
\begin{eqnarray}\label{eqn4}
h\left( {{\textbf{x}} + {\boldsymbol{\delta }}_{\textbf{x}} } \right) - h\left( {\textbf{x}} \right) =- {\mathop{\rm E}\nolimits}\left\{ {{\boldsymbol{\phi }}_{\textbf{x}}^T \left( {{\textbf{x}}} \right){\bf{\tilde X\boldsymbol \varepsilon }}} \right\} + o\left( \boldsymbol \varepsilon \right)
\end{eqnarray}
Similarly, for the random vector ${\bf{y}}$ we have,
\begin{eqnarray}\label{eqn6}
h\left( {{\textbf{y}} + {\boldsymbol{\delta }}_{\textbf{y}} } \right) - h\left( {\textbf{y}} \right) =- {\mathop{\rm E}\nolimits}\left\{ {{\boldsymbol{\phi }}_{\textbf{y}}^T \left( {\textbf{y}} \right){\bf{\tilde Y\boldsymbol \varepsilon }}}  \right\} + o\left( \boldsymbol \varepsilon \right)
\end{eqnarray}

Now, we want to obtain a similar expression for the last term of~\eqref{eqn3}. Using~\eqref{eqn10} and the multivariate version of Taylor series expansion for expanding $\ln p_{{\textbf{x,y}}} \left( \textit{\textbf x} ,\textit{\textbf y} \right)$ about $\left( \textit{\textbf x} ,\textit{\textbf y} \right)$,
we have,
\begin{align}\label{eqn302}
h\left( {{{\textbf x + \boldsymbol\delta }}_{\textbf{x}} ,{{\textbf y + \boldsymbol\delta }}_{\textbf{y}} } \right) - h\left( {{\textbf{x}},{\textbf{y}}} \right)& = -{\mathop{\rm E}\nolimits} \left\{ {\boldsymbol \phi _{{\textbf x, \textbf y}}^T\left( {\textbf x}, {\textbf y} \right)\left( {\left({\bf{\tilde X\boldsymbol \varepsilon }}\right)^T ,\left( {\bf{\tilde Y \boldsymbol\varepsilon }}\right)^T } \right)^T } \right\} + o\left( \boldsymbol \varepsilon \right) \notag \\
 \,\,\,\,\,& =- {\mathop{\rm E}\nolimits} \left\{ {{\boldsymbol \phi _{{\textbf{x,y}}}^{{\textbf{x}}\;\;T} \left( {\textbf x}, {\textbf y} \right)} {\bf{\tilde X \boldsymbol \varepsilon }}+{\boldsymbol \phi _{{\textbf{x,y}}}^{{\textbf{y}}\;\;T} \left( {\textbf x}, {\textbf y} \right)}{\bf{\tilde Y\boldsymbol \varepsilon }} }\right\} + o\left( \boldsymbol \varepsilon \right)
 \end{align}
Substituting~\eqref{eqn4},~\eqref{eqn6} and~\eqref{eqn302} in~\eqref{eqn3} completes the proof.
\subsection{Proof of Corollary~\ref{Corollary MAC2} }\label{Gaussian MAC}
From~\eqref{eqnMAC7} we can write,
\begin{align}\label{eqnMAC proof}
 \frac{{\partial I\left( {U_{\left( 1 \right)}^n ,...,U_{\left( K \right)}^n ;Y^n } \right)}}{{\partial \gamma _{\left( l \right)} }} &= \sum\limits_{i = 1}^n {{\mathop{\rm E}\nolimits} \left\{ {\left( {g_{_{\left( l \right)i} }  + \sum\limits_{k = 1}^K {\gamma _{\left( k \right)} \frac{{\partial g_{_{\left( k \right)i} } }}{{\partial \gamma _{\left( l \right)} }}} } \right)\frac{{\partial \ln p_{U_{\left( 1 \right)}^n ,...,U_{\left( K \right)}^n \left| {Y^n } \right.} \left( {U_{\left( 1 \right)}^n ,...,U_{\left( K \right)}^n \left| {Y^n } \right.} \right)}}{{\partial Y_i }}} \right\}}  \notag \\
  &= \sum\limits_{i = 1}^n {{\mathop{\rm E}\nolimits} \left\{ {\left( {g_{_{\left( l \right)i} }  + \sum\limits_{k = 1}^K {\gamma _{\left( k \right)} \frac{{\partial g_{_{\left( k \right)i} } }}{{\partial \gamma _{\left( l \right)} }}} } \right)\frac{{\partial \ln p_{U_{\left( 1 \right)}^n ,...,U_{\left( K \right)}^n ,Y^n }\left( {U_{\left( 1 \right)}^n ,...,U_{\left( K \right)}^n ,Y^n } \right)}}{{\partial Y_i }}} \right\}}  \notag \\
  &\,\,\,\,\,- \sum\limits_{i = 1}^n {{\mathop{\rm E}\nolimits} \left\{ {\left( {g_{_{\left( l \right)i} }  + \sum\limits_{k = 1}^K {\gamma _{\left( k \right)} \frac{{\partial g_{_{\left( k \right)i} } }}{{\partial \gamma _{\left( l \right)} }}} } \right)\frac{{\partial \ln p_{Y^n } \left( {Y^n } \right)}}{{\partial Y_i }}} \right\}}
\end{align}
For additive Gaussian noise $W_i$, the first term of~\eqref{eqnMAC proof} can be simplified as,
\begin{eqnarray}\label{eqnMAC proof 1}
\begin{array}{l}
\sum\limits_{i = 1}^n {{\mathop{\rm E}\nolimits} \left\{ {\left( {{g_{\left( l \right)i}} + \sum\limits_{k = 1}^K {\gamma_{ \left( k \right)}\frac{{\partial {g_{\left( k \right)i}}}}{{\partial \gamma_{\left( l \right)}}}} } \right)\frac{{\partial \ln {p_{U_{\left( 1 \right)}^n,...,U_{\left( K \right)}^n,{Y^n}}}\left( {U_{\left( 1 \right)}^n,...,U_{\left( K \right)}^n,{Y^n}} \right)}}{{\partial {Y_i}}}} \right\}} \\
 = \sum\limits_{i = 1}^n {{\mathop{\rm E}\nolimits} \left\{ {\left( {{g_{\left( l \right)i}} + \sum\limits_{k = 1}^K {\gamma_{\left( k \right)}\frac{{\partial {g_{\left( k \right)i}}}}{{\partial \gamma_{\left( l \right)}}}} } \right)\frac{{{{\partial {p_{{Y^n}|U_{\left( 1 \right)}^n,...,U_{\left( K \right)}^n}}\left( {{Y^n}|U_{\left( 1 \right)}^n,...,U_{\left( K \right)}^n} \right)} \mathord{\left/
 {\vphantom {{\partial {p_{{Y^n}|U_{\left( 1 \right)}^n,...,U_{\left( K \right)}^n}}\left( {{Y^n}|U_{\left( 1 \right)}^n,...,U_{\left( K \right)}^n} \right)} {\partial {Y_i}}}} \right.
 \kern-\nulldelimiterspace} {\partial {Y_i}}}}}{{{p_{{Y^n}|U_{\left( 1 \right)}^n,...,U_{\left( K \right)}^n}}\left( {{Y^n}|U_{\left( 1 \right)}^n,...,U_{\left( K \right)}^n} \right)}}} \right\}} \\
 =  - \sum\limits_{i = 1}^n {{\mathop{\rm E}\nolimits} \left\{ {\left( {{g_{\left( l \right)i}} + \sum\limits_{k = 1}^K {\gamma_{\left( k \right)}\frac{{\partial {g_{\left( k \right)i}}}}{{\partial \gamma_{\left( l \right)}}}} } \right)\left( {{Y_i} - \sum\limits_{k = 1}^K {\gamma_{\left( k \right)}{g_{\left( k \right)i}}} } \right)} \right\}} 
\end{array}
\end{eqnarray}
The second term can be expanded as,
\begin{eqnarray}\label{eqnMAC proof 2}
\begin{array}{l}
\sum\limits_{i = 1}^n {{\mathop{\rm E}\nolimits} \left\{ {\left( {{g_{\left( l \right)i}} + \sum\limits_{k = 1}^K {\gamma_{ \left( k \right)}\frac{{\partial {g_{\left( k \right)i}}}}{{\partial \gamma_{ \left( l \right)}}}} } \right)\frac{{\partial \ln {p_{{Y^n}}}\left( {{Y^n}} \right)}}{{\partial {Y_i}}}} \right\}} \\
 = \sum\limits_{i = 1}^n {{\mathop{\rm E}\nolimits} \left\{ {\left( {{g_{\left( l \right)i}} + \sum\limits_{k = 1}^K {\gamma_{\left( k \right)}\frac{{\partial {g_{\left( k \right)i}}}}{{\partial \gamma_{\left( l \right)}}}} } \right)\frac{{{{\partial {p_{{Y^n}}}\left( {{Y^n}} \right)} \mathord{\left/
 {\vphantom {{\partial {p_{{Y^n}}}\left( {{Y^n}} \right)} {\partial {Y_i}}}} \right.
 \kern-\nulldelimiterspace} {\partial {Y_i}}}}}{{{p_{{Y^n}}}\left( {{Y^n}} \right)}}} \right\}} \\
 = \sum\limits_{i = 1}^n {\int { \cdots \int {{p_{U_{\left( 1 \right)}^n,...,U_{\left( K \right)}^n|{Y^n}}}\left( {u_{\left( 1 \right)}^n,...,u_{\left( K \right)}^n|{y^n}} \right)\left( {{g_{\left( l \right)i}} + \sum\limits_{k = 1}^K {\gamma_{\left( k \right)}\frac{{\partial {g_{\left( k \right)i}}}}{{\partial \gamma_{\left( l \right)}}}} } \right){a_1}du_{\left( 1 \right)}^n \cdots du_{\left( k \right)}^nd{y^n}} } } \\
 = \sum\limits_{i = 1}^n {\int { \cdots \int {{p_{U_{\left( 1 \right)}^n,...,U_{\left( K \right)}^n,{Y^n}}}\left( {u_{\left( 1 \right)}^n,...,u_{\left( K \right)}^n,{y^n}} \right)\left( {{g_{\left( l \right)i}} + \sum\limits_{k = 1}^K {\gamma_{\left( k \right)}\frac{{\partial {g_{\left( k \right)i}}}}{{\partial \gamma_{\left( l \right)}}}} } \right){a_2}du_{\left( 1 \right)}^n \cdots du_{\left( k \right)}^nd{y^n}} } } \\
 =  - \sum\limits_{i = 1}^n {{\mathop{\rm E}\nolimits} \left\{ {\left( {{g_{\left( l \right)i}} + \sum\limits_{k = 1}^K {\gamma_{\left( k \right)}\frac{{\partial {g_{\left( k \right)i}}}}{{\partial \gamma_{\left( l \right)}}}} } \right)\left( {{Y_i} - \sum\limits_{k = 1}^K {\gamma_{\left( k \right)}{\mathop{\rm E}\nolimits} \left\{ {{g_{\left( k \right)i}}|{Y^n}} \right\}} } \right)} \right\}} 
\end{array}
\end{eqnarray}
where, 
\begin{eqnarray}\label{eqnMAC proof 2_1}
{a_1} = \int { \cdots \int {\frac{{{p_{{Y^n}|U_{\left( 1 \right)}^n,...,U_{\left( K \right)}^n}}\left( {{y^n}|u_{\left( 1 \right)}^n,...,u_{\left( K \right)}^n} \right)}}{{\partial {Y_i}}}{p_{U_{\left( 1 \right)}^n,...,U_{\left( K \right)}^n}}\left( {u_{\left( 1 \right)}^n,...,u_{\left( K \right)}^n} \right)du_{\left( 1 \right)}^n \cdots du_{\left( k \right)}^n} } ,
\end{eqnarray}
and
\begin{eqnarray}\label{eqnMAC proof 2_2}
{a_2} = \int { \cdots \int {{p_{U_{\left( 1 \right)}^n,...,U_{\left( K \right)}^n|{Y^n}}}\left( {u_{\left( 1 \right)}^n,...,u_{\left( K \right)}^n|{y^n}} \right)\left( {{Y_i} - \sum\limits_{k = 1}^K {\gamma_{ \left( k \right)}{g_{\left( k \right)i}}} } \right)du_{\left( 1 \right)}^n \cdots du_{\left( k \right)}^n} }. 
\end{eqnarray}

Substituting~\eqref{eqnMAC proof 1} and~\eqref{eqnMAC proof 2} in~\eqref{eqnMAC proof}, we get the desired result as,
\begin{align}\label{eqnMAC proof 3}
\frac{{\partial I\left( {U_{\left( 1 \right)}^n,...,U_{\left( K \right)}^n;{Y^n}} \right)}}{{\partial {\gamma _{\left( l \right)}}}} &= \sum\limits_{i = 1}^n {{\mathop{\rm E}\nolimits} \left\{ {\left( {{g_{\left( l \right)i}} + \sum\limits_{k = 1}^K {{\gamma _{\left( k \right)}}\frac{{\partial {g_{\left( k \right)i}}}}{{\partial {\gamma _{\left( l \right)}}}}} } \right)\left( {{Y_i} - \sum\limits_{k = 1}^K {{\gamma _{\left( k \right)}}{\mathop{\rm E}\nolimits} \left\{ {{g_{\left( k \right)i}}|{Y^n}} \right\}} } \right)} \right\}} \\ \notag
&\,\,\,\,\, - {\mathop{\rm E}\nolimits} \left\{ {\left( {{g_{\left( l \right)i}} + \sum\limits_{k = 1}^K {{\gamma _{\left( k \right)}}\frac{{\partial {g_{\left( k \right)i}}}}{{\partial {\gamma _{\left( l \right)}}}}} } \right)\left( {{Y_i} - \sum\limits_{k = 1}^K {{\gamma _{\left( k \right)}}{g_{\left( k \right)i}}} } \right)} \right\}\\ \notag
 &= {\gamma _{\left( l \right)}}\sum\limits_{i = 1}^n {{\mathop{\rm E}\nolimits} \left\{ {{{\left( {{g_{\left( l \right)i}} - {\mathop{\rm E}\nolimits} \left\{ {{g_{\left( l \right)i}}|{Y^n}} \right\}} \right)}^2}} \right\}} \\ \notag
&\,\,\,\,\, + \sum\limits_{k = 1,k \ne l}^K {{\gamma _{\left( k \right)}}\left( {\sum\limits_{i = 1}^n {{\mathop{\rm E}\nolimits} \left\{ {\left( {{g_{\left( l \right)i}} - {\mathop{\rm E}\nolimits} \left\{ {{g_{\left( l \right)i}}|{Y^n}} \right\}} \right)\left( {{g_{\left( k \right)i}} - {\mathop{\rm E}\nolimits} \left\{ {{g_{\left( k \right)i}}|{Y^n}} \right\}} \right)} \right\}} } \right)} \\ \notag
 &\,\,\,\,\,+ \sum\limits_{i = 1}^n {{\mathop{\rm E}\nolimits} \left\{ {\left( {\sum\limits_{k = 1}^K {{\gamma _{\left( k \right)}}\left( {{g_{\left( k \right)i}} - {\mathop{\rm E}\nolimits} \left\{ {{g_{\left( k \right)i}}|{Y^n}} \right\}} \right)} } \right)\left( {\sum\limits_{k = 1}^K {{\gamma _{\left( k \right)}}\frac{{\partial {g_{\left( k \right)i}}}}{{\partial {\gamma _{\left( l \right)}}}}} } \right)} \right\}} 
\end{align}
where the second equality follows from a similar argument as the last part of the proof of~\cite[Theorem 3.1.]{EI-MMSE} extended for the case of $K$ users. Hence, the proof of~\eqref{eqnMAC88} follows. Similarly, repeating some arguments like the above arguments for the case of $\gamma _{(l)}  = \gamma$ $\left(l = 1,...,K\right)$ the proof of~\eqref{eqnMAC8} easily follows as well.

\bibliography{Bib}
\bibliographystyle{IEEE}

\end{document}